\documentclass[12pt,draftcls,onecolumn]{IEEEtran}

\IEEEoverridecommandlockouts    
\usepackage{psfrag}
\usepackage{amsmath,amssymb,amsfonts,bbm}
\usepackage{latexsym}
\usepackage{graphicx}

\usepackage{colortbl}
\usepackage{fancyhdr}
\usepackage{epstopdf}
\usepackage{float}
\usepackage{subcaption}
\usepackage{hyperref}
\usepackage{booktabs}

\usepackage[usenames,dvipsnames,svgnames]{xcolor}

\usepackage{tikz}
\usepackage{pgfplots}
\usetikzlibrary{chains}
\usetikzlibrary{fit}
\usepgfplotslibrary{dateplot}
\usetikzlibrary{shapes.symbols,patterns, arrows} 

\newtheorem{theorem}{Theorem}
\newtheorem{definition}{Definition}
\newtheorem{proposition}{Proposition}
\newtheorem{lemma}{Lemma}
\newtheorem{corollary}{Corollary}

\newtheorem{remark}{Remark}

\newtheorem{standing}{Standing Assumption}

\IEEEoverridecommandlockouts

\usepackage{soul}
\usepackage[ruled]{algorithm2e}

\newcommand{\R}{\mathbb{R}}
\newcommand{\Z}{\mathbb{Z}}
\newcommand{\N}{\mathbb{N}}
\newcommand{\bbS}{\mathbb{S}}

\newcommand{\mc}{\mathcal}

\newcommand{\be}{\begin{equation}}
\newcommand{\ee}{\end{equation}}

\newcommand{\argmin}{\arg \min}
\newcommand{\bs}{\boldsymbol}

\begin{document}

\title{ {Decentralized Convergence to Nash Equilibria in Constrained Deterministic \\ Mean Field Control} }
\author{Sergio Grammatico, Francesca Parise, Marcello Colombino, and  John Lygeros
\thanks{The authors are with the Automatic Control Laboratory, ETH Zurich, Switzerland.  Research partially supported by the European Commission under project DYMASOS (FP7-ICT 611281) and by the Swiss National Science Foundation (grant 2-773337-12). The first three authors contributed equally as principal authors.
E-mail addresses: \{\texttt{grammatico}, \texttt{parisef}, \texttt{mcolombi}, \texttt{lygeros}\}\texttt{@control.ee.ethz.ch}.
}
}
\maketitle         

\begin{abstract}
This paper considers decentralized control and optimization methodologies for large populations of systems, consisting of several agents with different individual behaviors, constraints and interests, and affected by the aggregate behavior of the overall population. For such large-scale systems, the theory of aggregative and mean field games has been {established} and successfully applied in various scientific disciplines. 
While the existing literature addresses the case of unconstrained agents, we formulate deterministic mean field control problems in the presence of heterogeneous convex constraints for the individual agents, 
for instance arising from agents with linear dynamics subject to convex state and control constraints. 
We propose several model-free feedback iterations to compute in a decentralized fashion a mean field Nash equilibrium in the limit of infinite population size.
We apply our methods to the constrained linear quadratic deterministic mean field control problem and to the constrained mean field charging control problem for large populations of plug-in electric vehicles.
\end{abstract}

\section{Introduction}
Decentralized control and optimization in large populations of systems are of interest to various scientific disciplines, such as engineering, mathematics, social sciences, system biology and economics. A population of systems comprises {several} interacting heterogeneous agents, each with its own individual dynamic behavior and interest. {For the case of small/medium size populations,} such interactions can be analyzed via dynamic noncooperative game theory \cite{basar:olsder}. 

On the other hand, for large populations of systems the analytic solution of the game equations becomes computationally intractable. Aggregative and population games \cite{kukushkin:04, dubey:haimanko:zapechelnyuk:06, jensen:10, sandholm} represent a viable solution method to address {large population} problems where the behavior of each agent is affected by some aggregate effect of all the agents, rather than by specific one-to-one effects.
This feature attracts substantial research interest, indeed motivated by several relevant applications, including demand side management (DSM) for large populations of {prosumers} in smart grids~
\cite{mohsenian-rad:10, bagagiolo:bauso:14, chen:li:louie:vucetic:14, ma:hu:spanos:14}, charging coordination for large fleets of plug-in electric vehicles (PEVs) \cite{ma:callaway:hiskens:13, gan:topku:low:13, parise:colombino:grammatico:lygeros:14}, congestion control for networks of shared resources \cite{barrera:garcia:15}, synchronization of populations of coupled oscillators in power networks~\cite{yin:mehta:meyn:shanbhag:12,dorfler:bullo:14}.

Along these lines, Mean Field (MF) games have emerged as a methodology to study multi-agent coordination problems where 
{each individual agent is influenced by the statistical distribution of the population, and its contribution to the population distribution vanishes as the number of agents grows \cite{huang:caines:malhame:03,huang:caines:malhame:07, lasry:lions:07}. Specific research attention has been posed to MF setups where the effect of the population on each individual agent is given by a weighted average among the agents' strategies.} Unlike aggregative games, the distinctive feature of MF games is the emphasis on the limit of infinite population size, as this abstraction allows one to approximate the average population behavior based on its statistical properties only \cite{huang:caines:malhame:03,huang:caines:malhame:07, lasry:lions:07}. In the most general case, as the number of agents tends to infinity, the coupled interactions among the agents can be modeled mathematically via a system of two coupled Partial Differential Equations (PDEs), the Hamilton--Jacobi--Bellman (HJB) PDE for the optimal response of each individual agent \cite{huang:caines:malhame:03,huang:caines:malhame:07} and the Fokker--Planck--Kolmogorov (FPK) PDE for the dynamical evolution of the population distribution \cite{lasry:lions:07}. 
From the computational point of view, in the classical MF game setups, all the agents need information regarding the statistical properties of the population behavior to solve the MF equations in a decentralized fashion.


In this paper, we {consider} \textit{deterministic} MF games, as in \cite{bagagiolo:bauso:14, ma:callaway:hiskens:13, parise:colombino:grammatico:lygeros:14, bauso:pesenti:13}, 
{with an information structure for the agents which differs from the one of classical MF games. Specifically, we assume that the agents do not have access to the statistical properties of the population but, 
on the contrary, react optimally to a common external signal, which is broadcast by a central population coordinator}. 
{This information structure is typical of many large-scale multi-agent coordination problems, for instance in large fleets of PEVs \cite{ma:callaway:hiskens:13, gan:topku:low:13, parise:colombino:grammatico:lygeros:14}, DSM in smart grids \cite{mohsenian-rad:10, chen:li:louie:vucetic:14, ma:hu:spanos:14}, and congestion control \cite{barrera:garcia:15}. 
We then define the \textit{mean field control} problem as the task of designing an incentive signal that the central coordinator should broadcast so that the decentralized optimal responses of the agents satisfy some desired properties, in terms of the original deterministic MF game.} Contrary to the standard approach used to solve MF games, our MF control approach allow us to compute (almost) Nash equilibria for deterministic MF games in which the individual agents are subject to \textit{heterogeneous convex constraints}, for instance arising from different linear dynamics, convex state and input constraints. 
{Our motivation} comes from the fact that constrained systems arise naturally in almost all engineering applications, playing an active role in the agent behavior. 

In the presence of constraints, the optimal response of each agent is in general not known in closed form. To overcome this difficulty, we build on mathematical definitions and tools from convex analysis and operator theory \cite{bauschke:combettes, berinde}, establishing useful regularity properties of the mapping describing the aggregate population behavior.
We solve the constrained deterministic MF control problem via several specific feedback iterations and show convergence to {an incentive signal generating a MF equilibrium in a \textit{decentralized} fashion, 
making our methods scalable as the population size increases.
Analogously to \cite{ma:callaway:hiskens:13, parise:colombino:grammatico:lygeros:14, huang:caines:malhame:07, bauso:pesenti:13}, we seek convergence to a MF Nash equilibrium, that is, we focus on equilibria in which each agent has no interest to change its strategy, given the aggregate strategy of the others.

The  contributions of the paper are hence the following:
\begin{itemize}
\item We address the deterministic mean field control problem for populations of agents with heterogeneous convex constraints.

\item We show that the set of optimal responses to an incentive signal that is a fixed point of the population aggregation mapping gets arbitrarily close to a mean field Nash equilibrium, as the population size grows.

\item We show several regularity properties of the mappings arising in constrained deterministic mean field control problems.

\item We show that specific feedback iterations are suited to solve constrained deterministic mean field control problems with specific regularity.

\item  We apply our results to the constrained linear quadratic deterministic mean field control problem and to the constrained mean field charging control problem for large populations of plug-in electric vehicles, showing extensions to literature results.
\end{itemize}


The paper is structured as follows. Section \ref{sec:LQ-motivation} presents as a motivating example the LQ deterministic MF control problem for agents with linear dynamics, quadratic cost function, convex state and input constraints. 
Section \ref{sec:problem} shows the general deterministic MF control problem and the technical result about the approximation of a MF Nash equilibrium.
Section \ref{sec:quest} contains the main results, regarding some regularity properties of parametric convex programs arising in deterministic MF problems and the decentralized convergence to a MF Nash equilibrium of specific feedback iterations.   
Section \ref{sec:applications-extensions} discusses two applications of our technical results; it revises the constrained LQ deterministic MF control problem and presents the constrained MF charging problem for a large populations of heterogeneous PEVs. Section \ref{sec:conclusion} concludes the paper and highlights several possible extensions and applications.
Appendix \ref{app:operator-theory} presents some background definitions and results from operator theory; Appendix \ref{app:finite-horizon} justifies the use of finite-horizon formulations to approximate infinite-horizon discounted-cost ones; {Appendix \ref{app:proofs} contains all the proofs of the main results.}

\subsection*{Notation}
$\R$, $\R_{>0}$, $\R_{\geq 0}$ respectively denote the set of real, positive real, non-negative real numbers; $\N$ denotes the set of natural numbers; $\Z$ denotes the set of integer numbers; for $a, b \in \Z$, $a \leq b$, $\Z[a,b] := [a,b] \cap \Z$.
$A^\top \in \R^{m \times n}$ denotes the transpose of $A \in \R^{n \times m}$. Given vectors $x_1, \ldots, x_T \in \R^n$, $\left[ x_1; \cdots; x_T \right] \in \R^{nT}$ denotes $\left[ x_1^{\top}, \cdots, x_T^{\top} \right]^\top \in \R^{n T}$.
Given matrices $A_1, \ldots, A_M$, $\text{diag}\left( A_1, \ldots, A_M\right)$ denotes the block diagonal matrix with $A_1, \ldots, A_M$ in block diagonal positions.
With $\bbS^n$ we denote the set of symmetric $n \times n$ matrices; for a given $Q \in \R^{n \times n}$, {the notations} $Q \succ 0$ ($Q \succcurlyeq 0$) and $Q \in \bbS_{\succ 0}^n$ ($Q \in \bbS_{\succcurlyeq 0}^n$) denote that $Q$ is symmetric and has positive (non-negative) eigenvalues.
We denote by $\mathcal{H}_{Q}$, with $Q \succ 0$, the Hilbert space $\R^n$ with inner product $\langle \cdot, \cdot \rangle_{Q}: \R^n \times \R^n \rightarrow \R$ defined as $\langle x, y \rangle_{Q} := x^\top Q y$, and induced norm $\left\| \cdot \right\|_Q : \R^n \rightarrow \R_{\geq 0}$ defined as $\left\| x \right\|_Q := \sqrt{ x^\top Q x }$. 
A mapping $f: \R^n \rightarrow \R^n$ is Lipschitz in $\mathcal{H}_{Q}$ if there exists $L>0$ such that $\left\| f(x) - f(y)\right\|_Q \leq L \left\| x-y\right\|_Q$ for all $x,y \in \R^n$.
$\text{Id}: \R^n \rightarrow \R^n$ denotes the identity operator, $\textup{Id}(x) := x$ for all $x \in \R^n$. Every mentioned set $\mathcal{S} \subseteq \R^n$ is meant to be nonempty, unless explicitly stated.
The projection operator in $\mathcal{H}_{Q}$, $\text{Proj}_{\mathcal{C}}^{Q} : \R^n \rightarrow \mathcal{C} \subseteq \R^n $, is defined as $\text{Proj}_{\mathcal{C}}^{Q}(x) := \arg \min_{y \in \mathcal{C}} \left\|x-y \right\|_Q = \arg \min_{y \in \mathcal{C}} \left\|x-y \right\|^2_Q$. $I_n$ denotes the $n$-dimensional identity matrix; $\bs{0}$ denotes a matrix of all $0$s; $\bs{1}$ denotes a matrix/vector of all $1$s.
$A \otimes B$ denotes the Kronecker product between matrices $A$ and $B$. Given $\mathcal{S} \subseteq \R^n$, $A \in \R^{n \times n}$ and $b \in \R^n$, $A \mathcal{S} + b$ denotes the set $\{ A x+b \in \R^n \mid x \in \mathcal{S} \}$; hence given $\mathcal{S}^1, \ldots, \mathcal{S}^N \subseteq \R^n$ and $a_1, \ldots, a_N \in \R$, $\textstyle \frac{1}{N}\left( \sum_{i=1}^{N} a_i \mathcal{S}^i \right)  := \left\{ \frac{1}{N} \sum_{i=1}^{N} a_i x^i \in \R^n \mid x^i \in \mathcal{S}^i \ \forall i \in \Z[1,N] \right\}$.
The notation $\varepsilon_N = \mathcal{O}\left(1/N\right)$ denotes that there exists $c > 0$ such that $\lim_{N \rightarrow \infty} N \varepsilon_N = c $.

\section{Constrained linear quadratic deterministic mean field control as motivating example} \label{sec:LQ-motivation}

We start by considering a population of $N \in \N$ agents, where each agent $i \in \Z[1,N]$ has discrete-time linear dynamics
\begin{equation}
s_{t+1}^i = A^i s_t^i + B^i u_t^i,
\end{equation}
{where} $s^i \in \R^p$ {is} the state variable, $u^i \in \R^m$  {is} the input variable, and $A^i \in \R^{p \times p}$, $B^i \in \R^{p \times m}$. For each agent, we consider time-varying  {state and input} constraints
\begin{equation}
s_t^i \in \mc{S}_t^i, \ u_t^i \in \mc{U}_t^i
\end{equation}
for all $t \in \N$, where $\mc{S}_t^i \subset \R^p$ and $ \mc{U}_t^i \subset \R^m$ are convex compact sets.

Let us consider that each agent $i \in \Z[1,N]$ seeks a dynamical evolution that, given the initial state $s_0^i \in \R^p$, minimizes the finite-horizon cost function 
\begin{equation} \label{eq:Ji}
\textstyle   J\left( \bs{s}^i, \bs{u}^i,  \frac{1}{N} \sum_{j=1}^N \bs{s}^j \right) = 
\textstyle \sum_{t=0}^{T-1} \left\| s_{t+1}^i - \left( \eta + \frac{1}{N} \sum_{j=1}^N s_{t+1}^j \right) \right\|_{Q_{t+1}}^{2} \! + \left\| u_t^i\right\|_{R_t}^2
\end{equation}
where $Q_{t+1} \in \bbS_{\succ 0}^{p}$,  $R_t \in \bbS_{\succ 0}^{m}$ for all $t \in \Z[0,T-1]$, 
$\bs{s}^i = \left[ s_1^i; \ldots; s_T^i \right] \in \left( \mc{S}_1^i \times \cdots \times \mc{S}_{T}^i \right) \subset \R^{p T}$, 
$\bs{u}^i = \left[ u_0^i; \ldots; u_{T-1}^i \right] \in \left( \mc{U}_0^i \times \cdots \times \mc{U}_{T-1}^i \right) \subset \R^{m T}$ for all $i \in \Z[1,N]$, and $\eta \in \R^p$.

The cost function $J^i$ in \eqref{eq:Ji} is the sum of two cost terms, 
$\textstyle \left\| s_{t+1}^i - \left(  \eta + \frac{1}{N} \sum_{j=1}^N s_{t+1}^j  \right) \right\|^{2}$ and $\left\| u_t^i\right\|^2$; the first penalizes deviations from the average population behavior plus some constant offset $\eta$, while the second penalizes the control effort of the single agent. Note that the time-varying weights $\{Q_t\}_{t=1}^{T}$ and $\{R_t\}_{t=0}^{T-1}$ can also model an exponential cost-discount factor as in \cite[Equation 2.6]{huang:caines:malhame:07}, e.g., $Q_t = \lambda^t Q$ and $R_t = \lambda^t R$ for some $\lambda \in (0,1)$ and $Q, R \succ 0$.

We emphasize that the optimal decision of each agent $i \in \Z[1,N]$, that is, a feasible dynamical evolution $(\bs{s}^i, \bs{u}^i)$ minimizing the cost $J$ in \eqref{eq:Ji}, also depends on the decisions of all other agents through the average state $\frac{1}{N} \sum_{j=1}^N \bs{s}^j$ among the population. 
This feature results in an aggregative game \cite{kukushkin:04, dubey:haimanko:zapechelnyuk:06, jensen:10} among the population of agents, and specifically in a (deterministic) MF game, because the individual agent's state/decision depends on the mean population state \cite{ma:callaway:hiskens:13, bauso:pesenti:13}.

The constrained LQ deterministic MF control problem then consists in steering the {optimal responses} to a noncooperative equilibrium $\left\{ \left( \bar{\bs{s}}^i, \bar{\bs{u}}^i \right) \right\}_{i=1}^{N}$ {of the original MF game}, which satisfies the constraints and is convenient for all the individual noncooperative agents, via an appropriate incentive signal. To solve the MF control problem, we consider an algorithmic setup where a central coordinator, called virtual agent in \cite[Section IV.B]{huang:caines:malhame:07}, broadcasts a macroscopic incentive related to the average population state $\frac{1}{N} \sum_{j=1}^N \bs{s}^j $ to all the agents. In other words, the individual agents have no detailed information about every other agent, 
nor about their statistical distribution, but only react to the information broadcast by a central coordinator, which is somehow related to their aggregate behavior.

Formally, we define the {\textit{optimal response}} to a given reference vector $\bs{z} \in \R^{p \, T}$ of each agent $i \in \Z[1,N]$ as the solution to the following finite horizon optimal control problem:
\begin{equation} \label{eq:xi-ui-MFC}
\begin{array}{l}
\left( \bs{s}^{i \, \star}( \bs{z} ), \, \bs{u}^{i \, \star}(\bs{z}) \right) := \vspace{0.1cm}\\
\begin{array}{rll}
\displaystyle \argmin_{  (\bs{s}, \bs{u})  } & J\left( \bs{s}, \bs{u}, \bs{z} \right) & \\
\text{s.t. } & s_{t+1} = A^i s_t + B^i u_t, & \\
 & s_{t+1} \in \mathcal{S}_{t}^{i}, \ u_t \in \mathcal{U}_{t}^{i} & {\forall t \in \Z[0,T-1]}. \\
\end{array}
\end{array}
\end{equation}
We assume that the optimization problem \eqref{eq:xi-ui-MFC} is feasible for all agents $i \in \Z[1,N]$, that is, given the initial state $s_0^i \in \R^p$, we assume that there exists a control input sequence 
$\bs{u}^i = [u_0^i; \ldots; u_{T-1}^i] \in \mc{U}_0^i \times \cdots \times \mc{U}_{T-1}^{i}$ such that the sets $\{ \mathcal{S}_t^{i} \}_{t=1}^{T}$ are reachable at time steps $t = 1, \ldots, T$, respectively \cite[Chapter 6]{blanchini:miani}. This assumption can be checked by solving a convex feasibility problem; furthermore, the set of initial states $s_0^i$ such that \eqref{eq:xi-ui-MFC} is solvable can be computed by solving the feasibility problem parametrically in $s_0^i$.

We refer to \cite[Section III]{huang:caines:malhame:07} for the stochastic continuous-time infinite-horizon unconstrained counterpart of 
our linear quadratic (LQ) MF {game}. Here we focus on a discrete-time finite-horizon formulation to effectively address state and input constraints, by embedding them in finite-dimensional convex quadratic programs (QPs) that are efficiently solvable numerically.

Let us now rewrite the optimization problem in \eqref{eq:xi-ui-MFC} in the following compact form:
\begin{align}\label{eq:opt_problem-LQ}
\begin{split}
\displaystyle \bs{y}^{i \, \star}( \bs{z} ) := & \argmin_{  \bs{y} \in \bs{\mc{Y}}^i  }  J\left( \bs{y}, \bs{z} \right) \\
 = & \argmin_{  \bs{y} \in \bs{\mc{Y}}^i  }  \left\| \bs{y} \right\|_{ \bs{Q} }^2 + \left\|  \bs{y} - \left[\begin{smallmatrix} \bs{z} \\ \bs{0} \end{smallmatrix}\right] \right\|_{ \bs{\Delta} }^2 + 2 \bs{c}^\top \bs{y}
\end{split}
\end{align}
where $\bs{y} = \left[ \bs{s}; \bs{u} \right] \in \R^{(p+m)T}$, \\ $\bs{c} := -\textup{diag}\left( \textup{diag}\left( Q_1, \ldots, Q_T \right), \bs{0} \right) \left( \left[ \begin{smallmatrix} \bs{1} \\ \bs{0} \end{smallmatrix} \right] \otimes \eta  \right)$,
with
\begin{align} \label{eq:Qi-LQ}
\bs{Q} := & \ \textup{diag}\left( \bs{0}, \text{diag}\left( R_0, \ldots, R_{T-1} \right) \right), \\
\bs{\Delta} := & \ \textup{diag}\left( \text{diag}\left( Q_1, \ldots, Q_{T} \right), \bs{0} \right),
\end{align}
and, for a given initial condition $s_0^i \in \R^p$, 
\begin{equation} \label{eq:Yi-LQ}
\bs{\mc{Y}}^i := \left\{ \left[ \bs{s}; \bs{u} \right] \in \R^{(p+m)T} \mid \ s_{t+1} = A^i s_t + B^i u_t, \right. 
\left.  s_{t+1} \in \mathcal{S}_{t+1}^{i}, u_t \in \mathcal{U}_{t}^{i} \quad \forall t \in \Z[0,T-1] \right\}.
\end{equation} 

Motivated by the constrained LQ MF setup, in the next section we consider a broader class of deterministic MF control problems. In Section \ref{sec:Discrete-time constrained linear quadratic mean field control} we then apply the technical results in Sections \ref{sec:problem}, \ref{sec:quest} to the constrained LQ MF control problem. 


\section{Deterministic mean field control problem with convex constraints} \label{sec:problem}

\subsection{Constrained deterministic mean field game {with quadratic cost function}}

We consider a large population of $N \in \N$ heterogeneous agents, where each agent $i \in \Z[1,N]$ controls its decision variable $x^i$, taking values in the \textit{compact} and \textit{convex} set $\mathcal{X}^i \subset \R^n$. The aim of agent $i$ is to minimize its individual deterministic cost 
$ \textstyle J\left(x^i, \sigma \right),$ 
which depends on its own strategy $x^i$ and on the weighted average of strategies of all the agents, that is $\sigma:=\frac{1}{N} \sum_{j=1}^N  a_j x^j$ for some aggregation parameters $a_1, \ldots, a_N \geq 0$.  {Technically, each agent $i \in \Z[1,N]$ aims at computing the \textit{best response} to the other agents' strategies $\boldsymbol{x}^{-i} := ( x^1, \ldots, x^{i-1}, x^{i+1}, \ldots, x^{N} )$, that is,
\begin{align} \label{eq:br}
x^{i}_{ \textup{br} }(\boldsymbol{x}^{-i}) := & \arg \min_{ y \in \mathcal{X}^i } J\left(y, \frac{a_i}{N} y + \frac{1}{N}  \sum_{j\neq i}^{N} a_j {x}^j \right). 
\end{align}
Note that the best response mapping $x^{i}_{ \textup{br} }$ depends only on the aggregate of the other players strategies, thus leading to a MF game setup.
In classical game theory, a set of strategies in which every agent is playing a best response to the other players strategies is called Nash equilibrium. In the MF case, the concept is similar: if the population is at a MF Nash equilibrium, then each agent has no individual benefit to change its strategy, given the \textit{aggregation} among the strategies of the others.}

\vspace{0.2cm}
\begin{definition}[Mean field Nash equilibrium] \label{def:Nash-equ}
Given a cost function $J: \R^n \times \R^n \rightarrow \R$ and aggregation parameters $a_1, \ldots, a_N \geq 0$, a set of strategies $\{ \bar{x}^i \in \mc{X}^i \subseteq \R^{n} \}_{i=1}^{N} $ is a MF $\varepsilon$-Nash equilibrium, with $\varepsilon > 0$, if for all $i \in \Z[1,N]$ it holds
\begin{equation} \label{eq:epsilon-Nash}
 J\left( \bar{x}^i, \textstyle \frac{1}{N} \sum_{j=1}^{N} a_j \bar{x}^j \right) \leq 
\min_{ y \in \mathcal{X}^i } \textstyle J\left(y, \frac{1}{N} \left( a_i y + \sum_{j\neq i}^{N} a_j \bar{x}^j \right) \right) + \varepsilon.
\end{equation}
It is a MF Nash equilibrium if \eqref{eq:epsilon-Nash} holds with $\varepsilon=0$.
{\hfill $\square$}
\end{definition}
\vspace{0.2cm}

In the sequel, we consider the class of deterministic MF games with convex \textit{quadratic} cost $J: \R^n \times \R^n \rightarrow \R$ defined as
\begin{align} \label{cost}
J(x,\sigma) := \left\| x \right\|_Q^2 + \left\| x-\sigma \right\|_{\Delta}^{2} + 2 \left( C \sigma + c \right)^\top x
\end{align}
where $Q, \Delta \in \bbS_{\succcurlyeq 0}^{n}$, $Q+\Delta \in \bbS_{\succ 0}^{n}$, $C \in \R^{n \times n}$ and $c \in \R^n$. 

The three cost terms in \eqref{eq:constrained-optimizer} emphasize the contribution of three different contributions to the cost function: a quadratic cost $\left\| x \right\|_Q^2$, 
typical of LQ MF games \cite{bauso:pesenti:13, huang:caines:malhame:07, huang:caines:malhame:12}, 
a quadratic penalty $\left\| x-\sigma \right\|_{\Delta}^{2}$ on the deviations from the aggregate information \cite{huang:caines:malhame:07, ma:callaway:hiskens:13}, and an affine price-based incentive $2 \left( C \sigma + c \right)^\top x = p(\sigma)^\top x$ \cite{ma:callaway:hiskens:13, parise:colombino:grammatico:lygeros:14}. Let us also notice that he agents are fully heterogeneous relative to the constraint sets $\{ \mathcal{X}^i \}_{i=1}^N$.

Throughout the paper, we consider uniformly bounded aggregation parameters and individual constraint sets for all population sizes, which is typical of all the mentioned engineering applications.
\vspace{0.2cm}
\begin{standing}[Compactness] \label{ass:uniform-compactness}
There exist $\bar{a} >0$ and a compact set $\mathcal{X} \subset \R^n$
such that $a_i \in [0,\bar{a}]$ for all $i \in \Z[1,N]$, $\sum_{i=1}^{N} a_i = N$, and $\mathcal{X} \supseteq \cup_{i=1}^{N} \mathcal{X}^i$ hold for all $N$.
{\hfill $\square$}
\end{standing}

\vspace{0.2cm}
\begin{remark}
The formulation in \eqref{eq:constrained-optimizer} subsumes the one in \eqref{eq:opt_problem-LQ} with 
general $Q, \Delta$ in place of the $\bs{Q}, \bs{\Delta}$ from \eqref{eq:Qi-LQ}, general $C, c$ in place of $\bs{0}, \bs{c}$ in \eqref{eq:opt_problem-LQ}, general convex set $\mc{X}^i$ in place of $\bs{\mc{Y}}^i$ in \eqref{eq:Yi-LQ}, and any finite dimension $n \in \N$ in place of $(p+m)T$ in \eqref{eq:opt_problem-LQ}--\eqref{eq:Yi-LQ}. 
The notation $x^i$ replaces $\bs{y}^i$ in \eqref{eq:opt_problem-LQ} for the (state and input) decision strategy of each agent $i$.
{\hfill $\square$}
\end{remark}
\vspace{0.2cm}

Given {the cost function in \eqref{cost}} and the uniform boundedness assumption, the MF game in \eqref{eq:br} admits at least one MF Nash equilibrium, see Definition \ref{def:Nash-equ}, as stated next and proved in Appendix \ref{app:proofs}.
\vspace{0.2cm}
\begin{proposition} \label{prop:exists-Nash}
There exists a MF Nash equilibrium for the game in \eqref{eq:br} with $J$ as in \eqref{cost}.
{\hfill $\square$}
\end{proposition}
\vspace{0.2cm}

\subsection{Information structure and mean field control}
\label{sec:parametric-convex-programs}

We notice that to compute {the best response strategy} $x^i_{ \textup{br} }$ 
each agent $i$ would need to know the aggregation among the strategies of all other agents, namely $\frac{1}{N} \sum_{j\neq i}^{N}  a_j x^{j}$.
Motivated by several large-scale multi-agent applications \cite{mohsenian-rad:10, chen:li:louie:vucetic:14, ma:hu:spanos:14, ma:callaway:hiskens:13, parise:colombino:grammatico:lygeros:14, barrera:garcia:15}, 
here we consider a different information structure where each individual agent $i$ has neither knowledge about the states $\{ x^j \}_{j\neq i}$ of the other agents, nor about the aggregation parameters $\{ a_j \}_{j=1}^{N}$. {Instead, here every agent $i$ reacts to some macroscopic incentive, which is a function of the aggregate information about the whole population, including its contribution $x^i$ as well, and is broadcast to all the agents. 
Given this information structure,
we assume that each agent $i \in \Z[1,N]$ reacts to a broadcast signal $z \in \R^n$ through the optimal-response mapping} $x^{i \, \star}: \R^n \rightarrow \mc{X}^i \subset \R^n$ defined as
\begin{align} \label{eq:constrained-optimizer}
x^{i \, \star}(z) := & \arg \min_{ x \in \mathcal{X}^i } J(x,z) 
\end{align}
where $J$ is as defined in \eqref{cost}. 

{Moreover, let us formalize the aggregate (e.g., average) population behavior obtained when all the agents react optimally to a macroscopic signal by defining with the aggregation mapping 
$\mathcal{A}: \R^n \rightarrow \left( \frac{1}{N} \sum_{i=1}^{N} a_i \mathcal{X}^i \right) \subset \R^n$ as
\begin{equation} \label{eq:mean}
\textstyle \mathcal{A}(z) := \frac{1}{N} \sum_{i=1}^{N} a_i x^{i \, \star}(z).
\end{equation}}

{\begin{remark} 
The difference between the best response mapping $x^{i}_{ \textup{br} }$ defining the game in \eqref{eq:br}, and the optimal response mapping $x^{i \, \star}$ in \eqref{eq:constrained-optimizer} is that, while in the former an agent $i$ can also optimize its contribution $x^i$ in $\frac{1}{N} \sum_{j=1}^{N} a_j x^j $, in the latter the signal $z$ is fixed and hence the optimization in \eqref{eq:constrained-optimizer} is carried over the first argument of $J$ only.
{\hfill $\square$}
\end{remark}}
\vspace{0.2cm}

{According to the information structure described above, the MF control addresses the problem of designing a reference signal $\bar{z}$, such that the set of strategies $\{x^{i\,\star}(\bar z)\}_{i=1}^N$ possesses some desired properties, relative to the deterministic MF game in \eqref{eq:br}. Specifically, here we require the set of strategies to be an almost MF Nash equilibrium. 
To solve this MF control problem, we consider a} setup where the $N$ agents communicate to a central coordinator in a decentralized \textit{iterative} fashion. Namely, for a given broadcast signal $z_{(k)}$ at iteration $k \in \N$, each agent~$i$ computes its optimal response $x^{i \, \star}\left(z_{(k)} \right)$ based only on its own constraint set $\mathcal{X}^{i}$, that is its private information.
The central coordinator then receives the aggregate $\mathcal{A}\left(z_{(k)}\right)$ of all the individual responses, 
computes an updated reference $z_{(k+1)} = \Phi_k\left( z_{(k)}, \mathcal{A}\left(z_{(k)}\right) \right)$ through some feedback mapping $\Phi_k$,
broadcasts it to the whole population, 
and the process is repeated.

Technically speaking, given the cost function $J$, the agents constraint sets $\{ \mathcal{X}^i \}_{i=1}^{N}$ and the aggregation parameters $\{ a_i \}_{i=1}^{N}$,
the MF control problem consists in designing a signal $\bar{z} \in \R^n$, for instance 
via a feedback iteration $z_{(k+1)} = \Phi_k\left( z_{(k)}, \mathcal{A}\left(z_{(k)}\right) \right)$ such that, for any initial condition $z_{(0)} \in \R^n$,  $z_{(k)} \rightarrow \bar{z}$ which generates a MF (almost) Nash equilibrium $\{ x^{i \, \star}( \bar{z} ) \}_{i=1}^{N}$ {for the original MF game in \eqref{eq:br}}.

\subsection{Mean field Nash equilibrium in the limit of infinite population size} \label{sec:decentralized-convergence} 

Since the objective of our MF control problem is to find a MF Nash equilibrium for \textit{large population} size, we exploit the {Nash certainty equivalence principle or mean field approximation} idea \cite[Section IV.A]{huang:caines:malhame:07}. 
{Namely, for any agent $i$, the problem structure is such that the contribution of an individual strategy $x^i$ to the average population behavior $\mc{A}$ is negligible. 
Therefore, if $\bar{z} = \frac{1}{N} \sum_{j=1}^{N} a_j {x}^{j \, \star}(\bar{z}) = \mathcal{A}(\bar{z})$, then the optimal response ${x}^{i \, \star}(\bar{z})$ approximates the best response ${x}^{i}_{ \textup{br} }(\boldsymbol{x}^{-i}(\bar{z}))$ of agent~$i$ to the strategies $\{{x}^{j \, \star}(\bar{z})\}_{j\neq i}$ of the other players, for large population size $N$.}

Formally, under the uniform compactness condition for all population sizes in Standing Assumption \ref{ass:uniform-compactness},
the following result shows that a fixed point of the aggregation mapping $\mathcal{A}$ in \eqref{eq:mean} generates a MF Nash equilibrium in the limit of infinite population size.

\vspace{0.2cm}
\begin{theorem}[{Infinite population limit}] \label{thm:epsilon:nash}
\normalfont For all $\varepsilon > 0$, there exists $\bar{N}_{\varepsilon} \in \N$ such that, for all $N \geq \bar{N}_{\varepsilon}$, if $\bar{z}$ is a fixed point of $\mathcal{A}$ in \eqref{eq:mean}, that is,
$ \bar{z} = \frac{1}{N}\sum_{i=1}^{N} a_i x^{i \, \star}\left( \bar{z} \right)$, 
then the set
$\{ x^{i \, \star}\left( \bar{z} \right)\}_{i=1}^{N}$, 
with $x^{i \, \star}$ as in \eqref{eq:constrained-optimizer} for all $i \in \Z[1,N]$, is a MF $\varepsilon$-Nash equilibrium.
{\hfill $\square$}
\end{theorem}
\vspace{0.2cm}

\begin{remark} 
It follows from the proof of Theorem \ref{thm:epsilon:nash}, given in Appendix \ref{app:proofs}, that a fixed point of $\mathcal{A}$ in \eqref{eq:mean} with population size $N$ is a MF $\varepsilon_N$-Nash equilibrium with $\varepsilon_N = \mathcal{O}\left(1/N\right)$. 
Having a uniform upper bound $\bar{a}$ on the aggregation parameters $\{ a_i \}_{i=1}^{N}$ means that no single agent has a disproportionate influence on the population aggregation for large population size, which is a typical feature of MF setups \cite{bauso:pesenti:13, huang:caines:malhame:07, ma:callaway:hiskens:13}.
{\hfill $\square$}
\end{remark}
\vspace{0.2cm}

Theorem \ref{thm:epsilon:nash} suggests that we can design the feedback mappings $\left\{ \Phi_k \right\}_{k=1}^{\infty}$ to {iteratively steer the average population behavior} to a fixed point of the aggregation mapping $\mathcal{A}$ in \eqref{eq:mean}, as this is an approximate solution to the MF control problem for large population size.

\section{The quest for a fixed point of the aggregation mapping} \label{sec:quest}

\subsection{Mathematical tools from fixed point operator theory}
In this section we present the mathematical definitions needed for the technical results in Section~\ref{ref:main-results}, regarding appropriate fixed point iterations relative to the aggregation mapping. For ease of notation, the statements of this section are formulated in an arbitrary finite-dimensional Hilbert space $\mathcal{H}$, that is, in terms of an arbitrary norm $\left\| \cdot \right\|$ on $\R^n$, but in general hold for infinite-dimensional metric spaces.

We start from the property of contractiveness \cite[Definition 1.6]{berinde}, exploited in most of the MF control literature \cite{huang:caines:malhame:07, huang:caines:malhame:12, ma:callaway:hiskens:13} to show, under appropriate technical assumptions, convergence to a fixed point of the aggregation mapping.

\vspace{0.2cm}
\begin{definition}[{Contraction mapping}]
\label{def:CON}
A mapping $f: \R^n \rightarrow \R^n $ is a contraction (CON) if there exists $\epsilon \in (0,1]$ such that
\begin{equation} \label{eq:CON}
\left\| f(x) - f(y) \right\| \leq (1-\epsilon) \left\| x-y \right\|
\end{equation}
for all $x, y \in \R^n$.
{\hfill $\square$}
\end{definition}

If a mapping $f$ is CON, then the Picard--Banach iteration, $k = 0, 1, 2, \ldots$,
\begin{equation} \label{eq:Picard-Banach}
z_{(k+1)} = f \left( z_{(k)} \right) =: \Phi^{ \text{P--B} }\left( z_{(k)}, f \left( z_{(k)} \right) \right)
\end{equation}
converges, for any initial condition $z_{(0)} \in \R^n$, to its unique fixed point \cite[Theorem 2.1]{berinde}.


Although commonly used in the MF game literature \cite{huang:caines:malhame:07, huang:caines:malhame:12, ma:callaway:hiskens:13}, contractiveness is a quite restrictive property.
In this paper we actually exploit less restrictive properties than contractiveness, starting with nonexpansiveness \cite[Definition 4.1 (ii)]{bauschke:combettes}.

\vspace{0.2cm}
\begin{definition}[{NonExpansive mapping}]
A mapping $f: \R^n \rightarrow \R^n $ is nonexpansive (NE) if
\begin{equation} \label{eq:NE}
\left\| f(x) - f(y) \right\| \leq \left\| x-y \right\|
\end{equation}
for all $x, y \in \R^n$.
{\hfill $\square$}
\end{definition}
Clearly, a CON mapping is also NE, while the converse does not necessarily hold. Note that, unlike CON mappings, NE mappings, e.g., the identity mapping, may have more than one fixed point.
Among NE mappings, let us refer to firmly nonexpansive mappings \cite[Definition 4.1 (i)]{bauschke:combettes}. 
\vspace{0.2cm}
\begin{definition}[{Firmly NonExpansive mapping}] \label{def:FNE-mapping}
A mapping $f: \R^n \rightarrow \R^n $ is firmly nonexpansive (FNE) if
\begin{equation} \label{eq:FNE}
\left\| f(x) - f(y) \right\|^2 \leq \left\| x-y \right\|^2 - \left\|  f(x) - f(y) - \left( x-y\right) \right\|^2
\end{equation}
for all $x, y \in \R^n$.
{\hfill $\square$}
\end{definition}
An example of FNE mapping is the metric projection onto a closed convex set $\text{Proj}_{\mathcal{C}}: \R^n \rightarrow \mathcal{C} \subseteq \R^n$ \cite[Proposition 4.8]{bauschke:combettes}.

The FNE condition is sufficient for the Picard--Banach  in \eqref{eq:Picard-Banach} iteration to converge to a fixed point \cite[Section 1, p. 522]{combettes:pennanen:02}. 
This is not the case for NE mappings; for example, $z \mapsto f(z) := -z$ is NE, but not CON, and the Picard--Banach iteration $z_{(k+1)} = f(z_{(k)}) = -z_{(k)}$ oscillates indefinitively between $z_{(0)}$ and $-z_{(0)}$.
If a mapping $f: \mathcal{C} \rightarrow \mathcal{C}$ is NE, with $\mathcal{C} \subset \R^n$ compact and convex, then the Krasnoselskij iteration
\begin{equation} \label{eq:Krasnoselskij}
z_{(k+1)} = (1 - \lambda)z_{(k)} + \lambda  f \left( z_{(k)} \right) =: \Phi^{ \text{K} }\left(z_{(k)}, f \left( z_{(k)} \right)\right)
\end{equation}
where $\lambda \in (0,1)$, converges, for any initial condition $z_{(0)} \in \mathcal{C}$, to a fixed point of $f$ \cite[Theorem 3.2]{berinde}.

Finally, we consider the even weaker regularity property of  strict pseudocontractiveness \cite[Remark 4, pp. 12--13]{berinde}.
\vspace{0.2cm}
\begin{definition}[{Strictly PseudoContractive mapping}] 
\label{def:SPC}
A mapping $f: \R^n \rightarrow \R^n$ is strictly pseudocontractive (SPC) if there exists $\rho < 1$ such that
\begin{equation} \label{eq:SPC}
\left\| f(x) - f(y) \right\|^2 \leq \left\| x-y\right\|^2  + \rho \left\| f(x) - f(y) - \left(x-y \right)\right\|^2 
\end{equation}
for all $x, y \in \R^n$.
{\hfill $\square$}
\end{definition}

If a mapping $f: \mathcal{C} \rightarrow \mathcal{C}$ is SPC with $\mathcal{C} \subset \R^n$ compact and convex, then the Mann iteration
\begin{equation} \label{eq:Mann}
z_{(k+1)} =  (1-\alpha_k) z_{(k)} + \alpha_k f \left(z_{(k)} \right) =: \Phi_k^{ \text{M} }\left( z_{(k)}, f\left(z_{(k)} \right) \right) 
\end{equation}
where $\left(\alpha_k\right)_{k=0}^{\infty}$ is such that $\alpha_k \in (0,1) \ \forall k \geq 0$, $\lim_{k \rightarrow \infty} \alpha_k = 0$ and $\sum_{k=0}^{\infty} \alpha_k = \infty $, converges, for any initial condition $z_{(0)} \in \mathcal{C}$, to a fixed point of $f$ \cite[Fact 4.9, p. 112]{berinde}, \cite[Theorem R, Section I]{osilike:udomene:01}.

It follows from Definitions \ref{def:CON}--\ref{def:SPC} that $f$ FNE $\Longrightarrow$ $f$ NE, $f$ CON $\Longrightarrow$ $f$ NE $\Longrightarrow$ $f$ SPC. 
Therefore, the Mann iteration in \eqref{eq:Mann} ensures convergence to a fixed point for CON, FNE, NE and SPC mappings;
the Krasnoselskij iteration in \eqref{eq:Krasnoselskij} ensures convergence for CON, FNE and NE mappings; the Picard--Banach iteration in \eqref{eq:Picard-Banach}  for CON and FNE mappings.

The known upper bounds on the convergence rates suggest that a simpler iteration has faster convergence in general. The convergence rate for the Picard--Banach iteration is linear, that is  $\left\| z_{(k+1)} - \bar{z} \right\| / \left\| z_{(k)} - \bar{z} \right\| \leq 1 - \epsilon$ \cite[Chapter 1]{berinde}. Instead, the convergence rate for the Mann iteration is sublinear, specifically $\left\| z_{(k+1)} - \bar{z} \right\| / \left\| z_{(k)} - \bar{z} \right\| \leq 1 - \epsilon \, \alpha_{k}$ \cite[Chapter 4]{berinde}, for some $\epsilon > 0$.

Note that CON mappings have a unique fixed point \cite[Theorem 1.1]{berinde}, whereas FNE, NE, SPC mappings may have multiple fixed points. 
In our context this implies that, unless the aggregation mapping is CON, there could exist multiple MF Nash equilibria, which is effectively the case in multi-agent applications.

\subsection{Main results: Regularity and decentralized convergence} \label{ref:main-results}
Using the definitions and properties of the previous section, we can now state our technical result about the regularity of the optimal solution $x^{i \, \star}$ in \eqref{eq:constrained-optimizer} of the parametric convex program in \eqref{eq:constrained-optimizer}.

\vspace{0.2cm}
\begin{theorem}[Regularity of the optimizer] \label{th:mean-pseudocontractive}
Consider the following matrix inequality, where $Q, \Delta, C$ are from \eqref{eq:constrained-optimizer}:
\begin{equation} \label{eq:condition-contractive-constrained-optimizer}
\left[
\begin{matrix}
Q + \Delta & \Delta - C \\
\left( \Delta - C \right)^\top & Q + \Delta
\end{matrix} 
\right] \succcurlyeq \epsilon I.
\end{equation}
The mapping $x^{i \, \star}$ in \eqref{eq:constrained-optimizer} is: \vspace{0.1cm}\\
\begin{tabular}{llll}
CON & in $\mathcal{H}_{Q + \Delta}$ & if & \eqref{eq:condition-contractive-constrained-optimizer} holds with $\epsilon > 0$; \\
NE  &  in $\mathcal{H}_{Q + \Delta}$ & if & \eqref{eq:condition-contractive-constrained-optimizer} holds with $\epsilon \geq 0$; \\
FNE & in $\mathcal{H}_{\Delta-C}$ & if & $\Delta \succ C \succcurlyeq -Q$; \\
SPC & in $\mathcal{H}_{C-\Delta}$ & if & $\Delta \prec C$. \\
\end{tabular}

{\hfill $\square$}
\end{theorem}


\vspace{0.2cm}
\begin{remark}
The condition $\Delta \succ C \succcurlyeq -Q$ in Theorem \ref{th:mean-pseudocontractive} implies \eqref{eq:condition-contractive-constrained-optimizer} with $\epsilon = 0$, in fact
\begin{equation*}
\begin{array}{rcl}
\left[ \begin{smallmatrix} Q + \Delta & \Delta - C \\ \left( \Delta - C\right)^\top & Q + \Delta\end{smallmatrix} \right] & = &
I_2 \otimes \left( Q + C \right) + \bs{1} \otimes \left( \Delta - C \right)  \\
 & \succcurlyeq & \bs{1} \otimes \left( \Delta - C \right) \succcurlyeq 0,
\end{array}
\end{equation*}
where the last matrix inequality holds true because the eigenvalues of $\textstyle \bs{1} \otimes \left( \Delta - C \right) $ equal the product of the eigenvalues of $\Delta - C$, which are positive as 
$\Delta - C \succ 0$, and the eigenvalues of $\bs{1} = \left[\begin{smallmatrix}  1 & 1 \\ 1 & 1 \end{smallmatrix} \right]$, which are non-negative ($0$ and $2$).
{\hfill $\square$}
\end{remark}
\vspace{0.2cm}



We can now exploit the structure of the aggregation mapping $\mc{A}$ in \eqref{eq:mean} to establish our main result about its regularity. 
Specifically, under the conditions of Theorem \ref{th:mean-pseudocontractive}, the aggregation mapping  $\mc{A}$ inherits the same regularity properties of the individual optimizer mappings.

\vspace{0.2cm}
\begin{theorem}[Regularity of the aggregation] \label{prop:mean-pseudocontractive}
For all $i \in \Z[1,N]$, let $x^{i \, \star}$ be defined as in \eqref{eq:constrained-optimizer}. 
The mapping $\mathcal{A}$ in \eqref{eq:mean} is Lipschitz continuous, has a fixed point, and is: \vspace{0.1cm}\\ 
\begin{tabular}{llll}
CON & in $\mathcal{H}_{Q + \Delta}$ & if & \eqref{eq:condition-contractive-constrained-optimizer} holds with $\epsilon>0$; \\
NE   & in $\mathcal{H}_{Q + \Delta}$ & if & \eqref{eq:condition-contractive-constrained-optimizer} holds with $\epsilon \geq 0$; \\
FNE & in $\mathcal{H}_{\Delta-C}$ & if & $\Delta \succ C \succcurlyeq -Q$;  \\
SPC & in $\mathcal{H}_{C-\Delta}$ & if & $\Delta \prec C$. \\
\end{tabular}

{\hfill $\square$}
\end{theorem}


Theorem \ref{prop:mean-pseudocontractive} directly leads to iterative methods for finding a fixed point of the aggregation mapping, that is a solution of the MF control problem in the limit of infinite population size.

\vspace{0.2cm}
\begin{corollary}[Decentralized convergence] \label{cor:iterations}
The following iterations and conditions guarantee global convergence to a fixed point of $\mathcal{A}$ in \eqref{eq:mean}, where $x^{i \, \star}$ is as in \eqref{eq:constrained-optimizer} for all $i \in \Z[1,N]$: \vspace{0.1cm}\\
\begin{tabular}{lclll}
$1$. Picard--Banach &  \eqref{eq:Picard-Banach} & if & \eqref{eq:condition-contractive-constrained-optimizer} holds ($\epsilon > 0$) or \\
    & & & $\Delta \succ C \succcurlyeq -Q$; \\
$2$. Krasnoselskij &  \eqref{eq:Krasnoselskij} & if & \eqref{eq:condition-contractive-constrained-optimizer} holds ($\epsilon \geq 0$); & \\
$3$. Mann &  \eqref{eq:Mann} & if & \eqref{eq:condition-contractive-constrained-optimizer} holds ($\epsilon \geq 0$) or \\
 & & & $\Delta \prec C$.\\
\end{tabular}

{\hfill $\square$}
\end{corollary}

Note that convergence is ensured in different norms, namely $\textstyle \left\| \cdot \right\|_{ Q + \Delta }$, $\textstyle \left\| \cdot \right\|_{ \Delta - C}$ if $\Delta \succ C $ or $\textstyle \left\| \cdot \right\|_{ C- \Delta }$ if $C \succ \Delta$; this is not a limitation since all norms are equivalent in finite-dimensional Hilbert spaces.

We emphasize that each iterative method presented in Corollary \ref{cor:iterations} has its specific range of applicability depending on the specific MF control problem. This allows us to select one or more fixed point feedback iterations from the specific knowledge of the regularity property at hand. An important advantage of Corollary \ref{cor:iterations} is that decentralized convergence is guaranteed under conditions independent of the individual constraints $\{ \mathcal{X}^i \}_{i=1}^{N}$, but on the common cost function $J$ in \eqref{eq:constrained-optimizer} only. Therefore, our results and methods apply naturally to populations of heterogeneous agents.

Let us summarize in Algorithm \ref{alg:iterations} our proposed decentralized iteration to compute a fixed point of the aggregation mapping $\mathcal{A}$, where the feedback mapping $\Phi_k \in \left\{ \Phi^{\textup{P--B}}, \Phi^{\textup{K}}, \Phi_k^{\textup{M}} \right\}$ is chosen in view of Corollary \ref{cor:iterations}.

\begin{algorithm}[htb] 
\caption{Decentralized mean field control.}
\label{alg:iterations}

\textbf{Initialization}: $z \leftarrow z_{(0)}$, $k \leftarrow 0$.

\textbf{Iteration}:

$\quad$ $\displaystyle x^{i \, \star}(z) \leftarrow \arg \min_{x \in \mathcal{X}^i} J(x,z), \ i = 1, 2, \ldots, N$;
\vspace{0.1cm}

$\quad$ $ \mathcal{A}(z) \leftarrow \frac{1}{N} \sum_{i=1}^{N} a_i x^{i \star}(z)$;
\vspace{0.1cm}

$\quad$ $ z \leftarrow \Phi_k\left( z, \mathcal{A}(z)\right) $; \vspace{0.2cm}

$\quad$ $k \leftarrow k+1$.

\end{algorithm}

Note that under the conditions of Corollary \ref{cor:iterations}, Algorithm \ref{alg:iterations} guarantees convergence to a fixed point of the aggregation mapping $\mc{A}$ in \eqref{eq:mean} in a decentralized fashion. Let us also emphasize that any fixed point of $\mc{A}$ generates a MF $\varepsilon_N$-Nash equilibrium by Theorem \ref{thm:epsilon:nash}, that is not an exact Nash equilibrium for finite population size $N$, mainly because only some aggregate information $z$, which is related to $\frac{1}{N} \sum_{j=1}^N a_j x^j$, is broadcast to all the agents. In other words, we consider an information structure where each agent $i$ is not aware of the aggregate strategy $\frac{1}{N} \sum_{j \neq i}^N a_j x^j$ of the other agents $\{ x^j \}_{j \neq i}^N$, because this would require that, at each iteration step, the central coordinator communicates $N$ different quantities to the agents, namely $\frac{1}{N} \sum_{j=2}^N a_j x^j$ to agent $1$, $\frac{1}{N} \sum_{j \neq 2 }^N a_j x^j$ to agent $2$, up to $\frac{1}{N} \sum_{j =1 }^{N-1} a_j x^j$ to agent $N$.

\subsection{Discussion on decentralized convergence results in aggregative games}

Decentralized convergence to Nash equilibria in terms of fixed point iterations has been studied in aggregative game theory, for populations of finite size. 
Most of literature results show convergence of sequential (i.e., not simultaneous/parallel) best-response updates of the agents \cite[Cournot path]{kukushkin:04} \cite[Theorem 2]{jensen:10}, under the assumption that the best-response mappings of the players are non-increasing \cite[Assumption 1']{jensen:10}, besides continuous and compact valued.

In large-scale games, however, simultaneous/parallel responses as in Algorithm \ref{alg:iterations} are computationally more convenient with respect to sequential ones. Within the literature of aggregative games, the Mann iteration in \eqref{eq:Mann} has been proposed in \cite[Remark 2]{dubey:haimanko:zapechelnyuk:06} for the simultaneous (parallel) best responses of the agents. See \cite{heikkinen:06} for an application to distributed power allocation and scheduling in congested distributed networks. The aggregative game setup in these papers considers the strategy of the players to be a $1$-dimensional variable taking values in a compact interval of the real numbers. Convergence is then guaranteed if the best-response mappings of the players are continuous, compact valued and non-increasing \cite[conditions (i)--(iii), p. 81, Section 2]{dubey:haimanko:zapechelnyuk:06}.

It actually follows from the proof of Theorem \ref{prop:mean-pseudocontractive} that the condition $\Delta \prec C$ implies that the opposite  of aggregation mapping in \eqref{eq:mean}, i.e., $-\mc{A}(\cdot)$, is monotone, which is the $n$-dimensional generalization of the non-increasing property.
We conclude that Theorem \ref{prop:mean-pseudocontractive} provides mild sufficient conditions on the problem data such that the convergence result in Corollary \ref{cor:iterations} subsumes, limited to the quadratic cost function case, the one in \cite[Remark 2]{dubey:haimanko:zapechelnyuk:06}.

\section{Deterministic mean field control applications} \label{sec:applications-extensions}

\subsection{Solution to the constrained linear quadratic deterministic mean field control}
\label{sec:Discrete-time constrained linear quadratic mean field control}

In view of Theorem \ref{thm:epsilon:nash}, our discrete-time, finite-horizon, constrained LQ deterministic MF control problem from Section \ref{sec:LQ-motivation} reduces to finding a fixed point of the average mapping, that is, 
$\bs{z} \in \R^{pT} $ such that 
\begin{equation} \label{eq:fixed-point-mean}
\textstyle \bs{z} = \frac{1}{N} \sum_{i=1}^{N} \bs{x}^{i \, \star}( \bs{z} ) =: \mathcal{A}(\bs{z}),
\end{equation}
where $\bs{x}^{i \, \star}$ is defined in \eqref{eq:xi-ui-MFC}. In \eqref{eq:fixed-point-mean}, we average the optimal tracking trajectories $\{  \bs{x}^{i \, \star}( \bs{z} )\}_{i=1}^{N}$ among the whole population (that is, we take $a_1 = \cdots = a_N = 1$ in \eqref{eq:mean}, so that Assumption \ref{ass:uniform-compactness} is satisfied with $\bar{a} = 1$) and we require the trajectory $\bs{z}$ to equal such average. 
For large population size, 
the interpretation is that each agent $i$ responds optimally with state and control trajectory $\bs{x}^{i \, \star}( \bs{z} ) ,\, \bs{u}^{i \, \star}( \bs{z} )$, to the mass behavior $\bs{z} = \mathcal{A}(\bs{z})$ \cite[Section I, p. 1560]{huang:caines:malhame:07}.



In the \textit{unconstrained} linear quadratic setting, that is, $\mathcal{X}_{t}^{i} = \R^p$ and $\mathcal{U}_{t}^{i} = \R^m$ for all $i \in \Z[1,N]$ and $t \geq 0$, the mappings $\bs{x}^{i \, \star}$ and $ \bs{u}^{i \, \star}$ in \eqref{eq:xi-ui-MFC} are known in closed form, in both continuous- and discrete-time case, for both infinite and finite horizon \cite[Chapter 11]{anderson:moore}. 
Using this knowledge, if we replace $\textstyle \left( \eta + \frac{1}{N} \sum_{j=1}^{N} s_{t+1}^j \right)$ in \eqref{eq:Ji} by $\textstyle \gamma \left( \eta + \frac{1}{N} \sum_{j=1}^{N} s_{t+1}^j \right)$, for $\gamma \in \R$ small enough, then the corresponding mapping $\mathcal{A}$ from \eqref{eq:fixed-point-mean} is CON\footnote{If $\gamma=0$, then the mapping $\mathcal{A}$ in \eqref{eq:fixed-point-mean} is continuous, compact valued and constant, hence CON.} \cite[Theorem 3.4]{huang:caines:malhame:07}, and therefore the Picard--Banach iteration converges to the unique fixed point of $\mathcal{A}$ \cite[Proposition 3.4]{huang:caines:malhame:07}.


Unfortunately, it turns out that the mapping $\mathcal{A}$ in \eqref{eq:fixed-point-mean} is not necessarily CON.
We therefore apply the results in Section \ref{ref:main-results} to ensure convergence of suitable fixed point iterations.
Following \cite[Equation 2.6]{huang:caines:malhame:07}, for a given $\gamma \in \R$, let us consider
\begin{equation} \label{eq:J-x-u-z-g}
J_{\gamma}\left( \bs{s}, \bs{u}, \bs{z} \right) := \sum_{t=0}^{T-1} \left\| s_{t+1} - \gamma \left( \eta + z_{t+1}\right) \right\|_{Q_{t+1}}^2 + \left\| u_t \right\|_{R_t}^{2}
\end{equation}
which similarly to \eqref{eq:opt_problem-LQ} can be rewritten as a particular case of the cost function in \eqref{eq:constrained-optimizer} by choosing $\bs{Q}$ and $\bs{\Delta}$ as in \eqref{eq:Qi-LQ},
\begin{equation} \label{eq:mfc-rewritten-matrices}
\bs{C} := (1-\gamma) \bs{\Delta},
\end{equation}
$\bs{c} := -\gamma \, \text{diag}\left(\text{diag}\left( Q_1, \cdots, Q_T \right), \bs{0}\right) \left( \left[ \begin{smallmatrix} \bs{1} \\ \bs{0} \end{smallmatrix} \right] \otimes \eta \right)$.


Note that $\left( \bs{C} \left[ \bs{z} ; \bs{\tilde{z}}  \right] \right)^\top \bs{y} = (1-\gamma)  \left[ \bs{z} ; \bs{\tilde{z}}  \right] \bs{\Delta}^\top \bs{y} = (1-\gamma) \bs{z}^\top \text{diag}\left( Q_1, \cdots, Q_T \right) \bs{s}$ for all $\tilde{\bs{z}} \in \R^{mT}$, therefore $\tilde{\bs{z}}$ does not affect the optimization problem in \eqref{eq:opt_problem-LQ} with cost function $J_{\gamma}$ in \eqref{eq:J-x-u-z-g}. Here we formally consider a vector $\left[ \bs{z} ; \bs{\tilde{z}} \right]$ of the same dimensions of $\bs{y}$ just to recover the same mathematical setting in Section \ref{ref:main-results}.


We can now show conditions for the decentralized convergence to a fixed point of the average mapping in \eqref{eq:fixed-point-mean} for the discrete-time finite-horizon constrained LQ case, 
as corollary to our results in Section \ref{ref:main-results}.

\vspace{0.2cm}
\begin{corollary}[Fixed point iterations in LQ MF control]
\label{cor:Huang-converges}
The following iterations and conditions guarantee global convergence to a fixed point of $\mathcal{A}$ in \eqref{eq:fixed-point-mean}, where $\bs{s}^{i \, \star}$ is as in \eqref{eq:xi-ui-MFC} for all $i \in \Z[1,N]$, with $J_{\gamma}$ as in \eqref{eq:J-x-u-z-g} in place of $J$: \vspace{0.1cm}\\
\begin{tabular}{lcll}
$1$. Picard--Banach &  \eqref{eq:Picard-Banach} & if & $-1 <\gamma < 1$;\\
$2$. Krasnoselskij &  \eqref{eq:Krasnoselskij} & if & $-1 \le \gamma \leq 1$; \\
$3$. Mann &  \eqref{eq:Mann} & if & $-1 \le \gamma \leq 1$.
\end{tabular}

{\hfill $\square$}
\end{corollary}


\subsection{Production planning example}
Let us illustrate the LQ deterministic MF setting with a production planning example inspired by \cite[Section II.A]{huang:caines:malhame:07}. We consider $N$ firms supplying the same product to the market. Let $s_t^i \geq 0$ represent the production level of firm $i$ at time $t$. We assume that each firm can change its production according to the linear dynamics
\begin{align*}
s_{t+1}^{i} = s_t^i + u_t^i,
\end{align*}
where both the states and inputs are subject to heterogeneous constraints of the form $s_t^i \in [0,\bar s^i]$ and $ u_t^i \in [-\bar u^i,\bar u^i]$ for all $t \in \N$. 
We assume that the price of the product reads as
\begin{align*}
\textstyle p = p_0 - \rho \left(\frac{1}{N}\sum_{i=1}^N s^i\right),
\end{align*}
for $p_0, \rho >0$. 
Each firm seeks a production level $s^i$ proportional to the product price $p$, while facing the cost to change its production level (for example, for adding or removing production lines).
We can then formulate the associated LQ MF finite horizon cost function as
\begin{equation} \label{eq:J-x-u-z-numerical-example}
J\left( \bs{x}, \bs{u}, \bs{z} \right) := \sum_{t=0}^{T-1} \left ( s_{t+1} -   \gamma \left ( \eta + z_{t+1} \right) \right)^2 + r u_t ^2
\end{equation}
where $\eta := -p_0/\rho$, $\gamma:= - \rho$, $r> 0$, $\bs{s} = \left[ s_1, \ldots, s_T \right]^\top \in \R^{ T}$, $\bs{u} = \left[ u_0, \ldots, u_{T-1} \right]^\top \in \R^{T}$ and 
$\bs{z} = \left[ z_1, \ldots, z_T \right]^\top \in \R^{T}$.
Given a signal $\bs{z} \in \R^{T}$, each agent, $i = 1, \ldots, N$, solves a finite-horizon optimal tracking problem as defined in \eqref{eq:xi-ui-MFC}, with cost function $J$ in \eqref{eq:J-x-u-z-numerical-example}. 
For illustration, we consider the case of a heterogeneous population of firms where we randomly sample the upper bound $\bar{s}^i$ from a uniform distribution supported on $[0,\, 10]$ and $\bar{u}^i$ from a uniform distribution supported on $[0,\, \bar{s}^i/5]$. 
We consider the parameters $p_0=10$, $\rho = 1$, $T=20$, and hence $\gamma = -1$. The mapping $\mc{A}$ defined in \eqref{eq:fixed-point-mean} is then NE, thus the Krasnoleskij iteration in \eqref{eq:Krasnoselskij}  does guarantee convergence to a fixed point, according to Corollary \ref{cor:Huang-converges}.

For different population sizes $N$, we first numerically compute a fixed point $\bar{\bs{z}}$ of $\mc{A}$ using the Krasnoleskij iteration in \eqref{eq:Krasnoselskij} with parameter $\lambda=0.5$, and we hence compute the strategies 
$\left\{ \left( {\bs{s}}^{ i \, \star }( \bar{\bs{z}} ), { \bs{u}}^{ i \, \star }( \bar{\bs{z}} )  \right) \right\}_{i=1}^{N}$. We then verify that this is an $\varepsilon_N$-Nash equilibrium: for each firm $i$, we evaluate the individual cost $\bar{J}^i := J\left( \bs{s}^{ i \, \star }( \bar{\bs{z}} ), {\bs{u}}^{ i \, \star }( \bar{\bs{z}} ), \bar{\bs{z}} \right)$ and the actual optimal cost $J^{ i \, \star }$ under the knowledge of the production plan of the other firms at the fixed point $\bar{\bs{z}}$. 
In Figure \ref{fig:epsilon} we plot the maximum benefit $\varepsilon_N := \max_{ i \in \Z[1,N] } | J^{ i \, \star } - \bar{J}^i |$ that a firm could achieve by unilaterally deviating from the solution computed via the fixed point iteration, normalized by the optimal cost in the homogeneous case with expected constraints ($s^i \in [0,5]$, $u^i \in [-1, 1]$ $\, \forall i \in \Z[1,N]$). According to Theorem \ref{thm:epsilon:nash}, such benefit vanishes as the population size increases.
\begin{figure}[htp!]
\begin{center}
\includegraphics{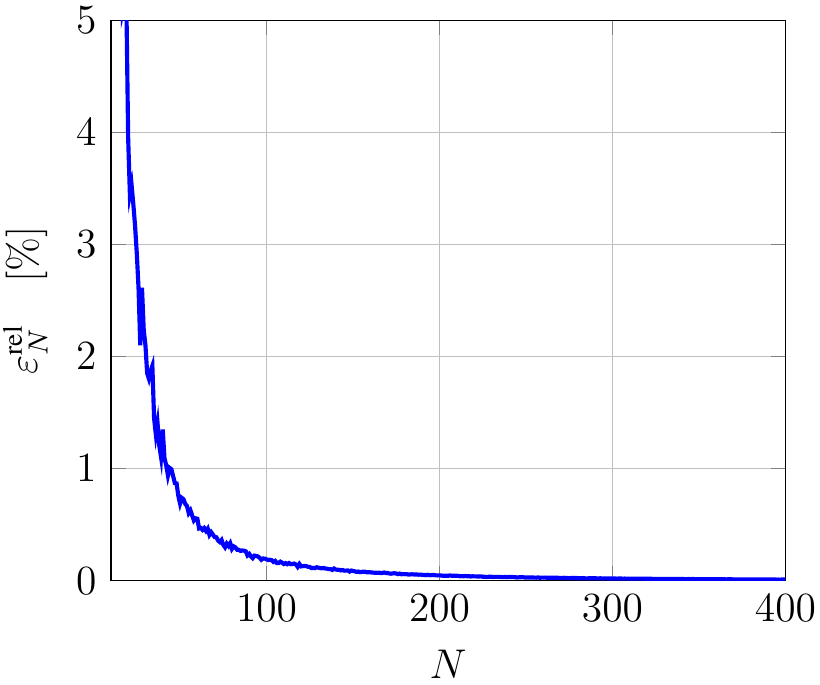}
\caption{As the population size $N$ increases, the maximum achievable individual cost improvement $\varepsilon_N$, relative to the optimal cost in the homogeneous case with expected constraints ($s^i \in [0,5]$, $u^i \in [-1, 1]$ $\, \forall i \in \Z[1,N]$), decreases to zero. For all population sizes, $N$ agents are randomly selected.}
\label{fig:epsilon}
\end{center}
\end{figure}

\subsection{Decentralized constrained charging control for large populations of plug-in electric vehicles} \label{sec:application-pevs}
As second control application, we investigate the problem of coordinating the charging of a large population of PEVs, introduced in \cite{ma:callaway:hiskens:13} and extended to the constrained case in \cite{parise:colombino:grammatico:lygeros:14}. 
For each PEV $i \in \Z[1,N]$, we consider the discrete-time, $t \in \N$, linear dynamics
\begin{equation*}
s_{t+1}^{i} = s_t^i + b^i u_t^i
\end{equation*}
where $s^i \in [0,1]$ is the state of charge, $u^i \in [0,1]$ is the charging control input and $b^i >0$ represents the charging efficiency.

The objective of each PEV $i$ is to acquire a charge amount $\gamma^i \in [0,1]$ within a finite charging horizon $T \in \N$, hence to satisfy the charging constraint\footnote{We could also consider more general convex constraints, for instance on the desired state of charge, multiple charging intervals, charging rates, vehicle-to-grid operations. However, we prefer to keep the same setting of \cite{ma:callaway:hiskens:13, parise:colombino:grammatico:lygeros:14} for simplicity.} $\sum_{t=0}^{T-1} u_t^i = \bs{1}^\top \bs{u}^i = \gamma^i$, 
while minimizing its charging cost $\sum_{t=0}^{T-1} p_t\left( \cdot \right) u_t^i = \bs{p}\left( \bs{\cdot} \right)^\top \bs{u}^i$, where $\bs{p}( \bs{\cdot} )^\top = \left[ p_0(\cdot), \ldots, p_{T-1}(\cdot)\right]^\top$ is the electricity price function over the charging horizon. 
We consider a dynamic pricing, where the price of electricity depends on the overall demand, namely the inflexible demand plus the aggregate PEV demand. In particular, in line with the (almost-affine) price function in \cite{ma:callaway:hiskens:13, parise:colombino:grammatico:lygeros:14},
we consider an affine price function $\bs{p}( \bs{z} ) := 2\left( a \bs{z} + \bs{c} \right)$, 
where $a>0$ represents the inverse of the price elasticity of demand and $\bs{c} \geq \bs{0}$ denotes the average inflexible demand.
The interest of each agent is to minimize its own charging cost $2( a \bs{z} + \bs{c} )^\top \bs{u}^i$, which however leads to a linear program with undesired discontinuous optimal solution.
Therefore, following \cite{ma:callaway:hiskens:13, parise:colombino:grammatico:lygeros:14}, we also introduce a quadratic relaxation term as follows.

The optimal charging control $\bs{u}^{i \, \star}$ of each PEV $i \in \Z[1,N]$, given the price signal $ \bs{z} = \left[ z_0, \ldots, z_{T-1} \right] \in \R^T$, is defined as
\begin{equation} \label{eq:optimal-charging-control}
\begin{array}{rrl}
\displaystyle \bs{u}^{i \, \star}( \bs{z} ) := & \displaystyle \arg \min_{ \bs{u} \in \R^T }  & \delta \left\| \bs{u} - \bs{z}   \right\|^2 + 2( a \bs{z} + \bs{c} )^\top \bs{u} \\ 
& \text{s.t. } & \bs{0} \leq \bs{u} \leq \bs{U}^i, \ \bs{1}^{\top} \bs{u} = \gamma_i,
\end{array} 
\end{equation}
where $\delta > 0$ and $\bs{U}^i \in \R_{\geq 0}^T$ is a vector of desired upper bounds on the charging inputs.
Note that the perturbation $\delta>0$ should be chosen small to approximate the original linear cost $2( a \bs{z} + \bs{c} )^\top \bs{u}^i$. We refer to \cite[Section V]{parise:colombino:grammatico:lygeros:14} for a numerical evidence of the beneficial effect of choosing a small $\delta>0$ for the perturbed cost in \eqref{eq:optimal-charging-control}. 

In view of Theorem \ref{thm:epsilon:nash}, a solution to the corresponding MF control problem is a fixed point of the mapping
\begin{equation} \label{eq:mean-control}
\textstyle \mathcal{A}( \bs{z} ) := \frac{1}{N} \sum_{i=1}^{N} \bs{u}^{i \, \star}( \bs{z} )
\end{equation}
which represents the average among the optimal charging control inputs $\{ \bs{u}^{i \, \star}( \bs{z} ) \}_{i=1}^{N}$.



Since the cost function in \eqref{eq:optimal-charging-control} is a particular case of the general cost function in \eqref{eq:constrained-optimizer}, namely with $Q=0$, $\Delta = \delta I$, $C = a I$, we can establish conditions on $\delta > 0$ under which a specific fixed point iteration, that in this context represents a price update feedback law, converges to a MF almost-Nash solution for the constrained charging control problem. In particular, the Mann iteration in \eqref{eq:Mann} always converges to a fixed point of the aggregation mapping and hence solves the constrained deterministic MF control problem in the limit of infinite population size.

\vspace{0.2cm}
\begin{corollary}[Decentralized PEV charging control]
\label{cor:Callaway-converges}
The following iterations and conditions guarantee global convergence to a fixed point of $\mathcal{A}$ in \eqref{eq:mean-control}, where $\bs{u}^{i \, \star}$ is as in \eqref{eq:optimal-charging-control} for all $i \in \Z[1,N]$: \vspace{0.1cm} \\
\begin{tabular}{lcll}
$1$. Picard--Banach &  \eqref{eq:Picard-Banach} & if & $\delta > a/2$;\\
$2$. Krasnoselskij &  \eqref{eq:Krasnoselskij} & if & $\delta \geq a/2$;\\
$3$. Mann &  \eqref{eq:Mann} & if & $\delta>0$. \\
\end{tabular}

{\hfill $\square$}
\end{corollary}

In \cite{ma:callaway:hiskens:13}, only the Picard--Banach iteration is considered, for some values of $\delta>a/2$.
For small values of $\delta$, it is shown in both \cite{ma:callaway:hiskens:13} and \cite{parise:colombino:grammatico:lygeros:14} that the Picard--Banach iteration causes permanent price oscillations. On the other hand, in \cite{parise:colombino:grammatico:lygeros:14} it is observed in simulation that the Mann iteration does converge. Corollary \ref{cor:Callaway-converges} hence provides theoretical support for this observation. 

Using the same numerical values as in \cite{ma:callaway:hiskens:13}, Figure \ref{fig:valley_filling_1} shows that, if we choose the parameter $\delta > 0$ small enough, we recover the valley-filling solution, desirable in the case without charging upper bounds \cite[Lemma 3.1]{ma:callaway:hiskens:13}. For the same case, we show in Figure \ref{fig:valley_filling_2} that the Picard--Banach iteration oscillates indefinitely, while the Mann iteration converges.
\begin{figure}[h!]
\begin{center}
\includegraphics{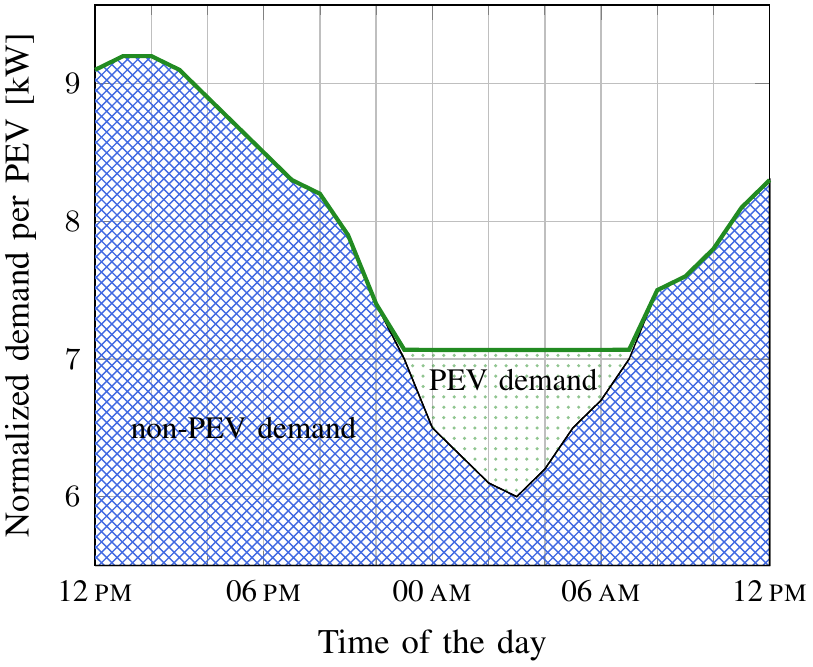}
\caption{Charging setting without upper bounds ($\delta=10^{-4}$): the Mann iteration  converges to a desirable valley-filling solution.}
\label{fig:valley_filling_1}
\end{center}
\end{figure}

\begin{figure}[h!]
\begin{center}
\includegraphics{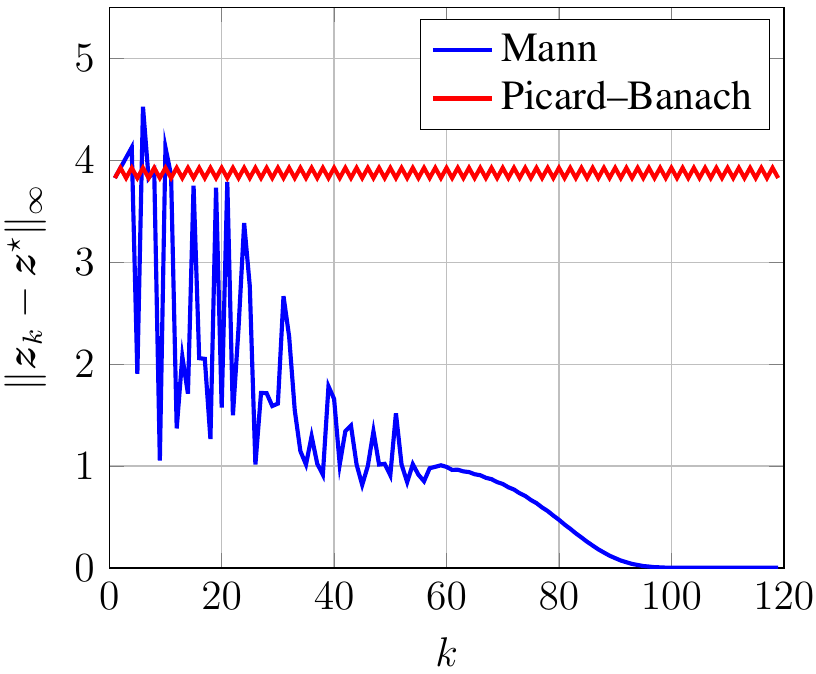}
\caption{Charging setting without upper bounds ($\delta=10^{-4}$): the Picard--Banach iteration  oscillates in a limit cycle while the Mann iteration converges to a desirable valley-filling solution $\bs{z}^\star$.}
\label{fig:valley_filling_2}
\end{center}
\end{figure}

We refer to \cite{parise:colombino:grammatico:lygeros:14} for further discussions and numerical simulations. Application to realistic PEV case studies is topic of current work.

\section{Conclusion and Outlook} \label{sec:conclusion}
\subsection*{Conclusion}

We have considered mean field control approaches for large populations of systems, consisting of agents with different individual behaviors, constraints and interests, and affected by the aggregate behavior of the overall population. We have addressed mean field control theory for problems with heterogeneous convex constraints, for instance arising from agents with linear dynamics subject to convex state and control constraints. We have proposed several model-free decentralized feedback iterations for constrained mean field control problems, as summarized in Table \ref{table:Table}, 
with guaranteed global convergence to a mean field Nash equilibrium for large population sizes.
We believe that our methods and results open several research directions in mean field control theory and inspire novel methods to various applications.

\setlength{\tabcolsep}{4pt}
\renewcommand{\arraystretch}{1.2}

\begin{table}
\caption{Conditions on the problem data, corresponding regularity properties of the aggregation mapping and iterations that ensure convergence to a fixed point of the aggregation mapping.}
\label{table:Table}
\begin{center}
\begin{tabular}{ @{} cc ccc @{} } 
\bottomrule
\multicolumn{5}{>{\columncolor[gray]{.95}}l}{Constrained deterministic MF control with quadratic cost (Sections \ref{sec:problem}, \ref{sec:quest})} \\
\toprule
 & & \multicolumn{3}{c}{Feedback iterations} \\
\cmidrule{3-5}
Condition & Property& Picard--Banach & Krasnoselskij  & Mann  \\
\toprule 
$ \left[ \begin{smallmatrix} Q+\Delta & \Delta-C \\ ( \Delta - C )^\top & Q + \Delta  \end{smallmatrix} \right] \succ 0 $  & CON & $\checkmark$ & $\checkmark$ & $\checkmark$ \vspace{0.1cm}\\
$ \Delta \succ C \succcurlyeq -Q $ & FNE  & $\checkmark$ & $\checkmark$ & $\checkmark$ \vspace{0.1cm} \\
$ \left[ \begin{smallmatrix} Q+\Delta & \Delta-C \\ ( \Delta - C )^\top & Q + \Delta  \end{smallmatrix} \right] \succcurlyeq 0 $  & NE &  & $\checkmark$ & $\checkmark$ \vspace{0.1cm} \\
$\Delta \prec C$ & SPC &  &  & $\checkmark$ \\
\toprule
\bottomrule
\multicolumn{5}{>{\columncolor[gray]{.95}}l}{Constrained LQ deterministic MF control (Sections \ref{sec:LQ-motivation}, \ref{sec:Discrete-time constrained linear quadratic mean field control})}\\
\toprule
 & & \multicolumn{3}{c}{Feedback iterations} \\
\cmidrule{3-5}
Condition & Property& Picard--Banach & Krasnoselskij  & Mann  \\
\toprule 
$ -1 < \gamma < 1  $     & CON & $\checkmark$ & $\checkmark$ & $\checkmark$ \\
$ -1 \leq \gamma \leq 1$ & NE    &   & $\checkmark$ & $\checkmark$ \\
\toprule
\bottomrule
\multicolumn{5}{>{\columncolor[gray]{.95}}l}{Constrained MF PEV charging control (Section \ref{sec:application-pevs})} \\
\toprule
 & & \multicolumn{3}{c}{Feedback iterations} \\
\cmidrule{3-5}
Condition & Property& Picard--Banach & Krasnoselskij  & Mann  \\
\toprule 
$ \delta > a/2$  & CON & $\checkmark$ & $\checkmark$ & $\checkmark$ \\
$ \delta \geq a/2$ & NE  &   & $\checkmark$ & $\checkmark$ \\
$ \delta > 0 $ & SPC &  &  & $\checkmark$ \\
\bottomrule

\end{tabular}
\end{center}
\end{table}

\setlength{\tabcolsep}{6pt}
\renewcommand{\arraystretch}{1}

\subsection*{Outlook on extensions and applications}

Most of the mathematical results from operator theory we adopted for finite-dimensional Euclidean spaces, also hold for infinite-dimensional Hilbert spaces. Therefore, our technical results can be potentially extended to infinite-horizon MF control problems.

We have considered agents with homogeneous cost functions, coupled via the aggregate population behavior. 
The cases of heterogeneous cost functions and couplings in the constraints are possible generalizations,
motivated by setups where different agents may have different local interests and local mutual constraints. 
Since we have considered agents with a strictly-convex quadratic cost function, a valuable generalization would be the case of general convex cost function.

As we have addressed a deterministic setting, inspired by the deterministic agent dynamics in \cite{ma:callaway:hiskens:13, bauso:pesenti:13}, a valuable extension would be a stochastic setting in the presence of state and input constraints. For instance, the parameters of each agent can be thought as extracted from a probability distribution \cite[Section V]{huang:caines:malhame:07}, and/or a zero-mean random input can enter linearly in the dynamics \cite[Equation 2.1]{huang:caines:malhame:07}.

The concept of social global optimality has not been considered in this paper. Following the lines of \cite[Section IV]{huang:caines:malhame:12}, it would be valuable to show, under suitable technical conditions, that the MF structure allows one to coordinate efficiently decentralized constrained optimization schemes.

Our constrained MF setup can be also extended in many transverse directions. For instance, the effect of local heterogeneous constraints can be studied in MF games with leader-follower (major-minor) agents \cite{nourian:caines:malhame:huang:12} and in coalition formation MF games \cite{kizilkale:caines:12}.
Furthermore, we believe that our constrained setting and methods can be also exploited in network games with local interactions \cite{bauso:giarre:pesenti:08, huang:caines:malhame:10}.

Applications of our methods and results include decentralized control and game-theoretic coordination in large-scale systems. Among others, application domains that can be further explored in view of our constrained MF setup are dynamic demand-side management of aggregated loads in power grids
\cite{mohsenian-rad:10, chen:li:louie:vucetic:14, bagagiolo:bauso:14, ma:hu:spanos:14}, congestion control over networks \cite{barrera:garcia:15}, synchronization and frequency regulation among populations of coupled oscillators \cite{yin:mehta:meyn:shanbhag:12, dorfler:bullo:14}.
Another application field suited for our constrained MF control approach is the supply-demand regulation in energy markets, where agents with heterogeneous behaviors and interests, wish to efficiently buy and/or sell services and energy \cite{kizilkale:mannor:caines:12}. 


%
%
%

\section*{Acknoledgements}
The authors would like to thank Basilio Gentile for fruitful discussions on the topic. 

\appendix

\subsection{Further mathematical tools from operator theory} \label{app:operator-theory}
In this section, we present some useful operator theory definitions, adapted to finite-dimensional Hilbert spaces from \cite{bauschke:combettes, berinde}.
For completeness, we present the most general known fixed point iteration, that is, the Ishikawa iteration in \eqref{eq:Ishikawa}, which guarantees convergence to a fixed point of a (non-strictly) PseudoContractive (PC) mapping \cite[Theorem 5.1]{berinde}, as formalized next \cite[Remark 3, pp. 12--13]{berinde}.

\vspace{0.2cm}
\begin{definition}[{PseudoContractive mapping}] \label{def:PC}
A mapping $f: \R^n \rightarrow \R^n$ is pseudocontractive (PC) in $\mathcal{H}_P$ if
\begin{equation}
\left\| f(x) - f(y) \right\|_P^2 \leq \left\| x-y\right\|_P^2  + \left\| f(x) - f(y) - \left(x-y \right)\right\|_P^2 
\end{equation}
for all $x, y \in \R^n$.
{\hfill $\square$}
\end{definition}

If a mapping $f: \mathcal{C} \rightarrow \mathcal{C}$ is PC and Lipschitz in $\mathcal{H}_P$, with $\mathcal{C} \subseteq \R^n$ compact and convex, then the Ishikawa iteration
\begin{equation}\label{eq:Ishikawa}
z_{(k+1)} =  (1-\alpha_k) z_{(k)} + \alpha_k f \left( (1-\beta_k) z_{(k)} + \beta_k f\left( z_{(k)} \right) \right)
\end{equation}
where $\left(\alpha_k\right)_{k=0}^{\infty}$, $\left(\beta_k\right)_{k=0}^{\infty}$ are such that $0 \leq \alpha_k \leq \beta_k \leq 1 \ \forall k \geq 0$, $\lim_{k\rightarrow \infty} \beta_k = 0$ and $\sum_{k=0}^{\infty} \alpha_k \beta_k = \infty $, converges, for any initial condition $z_{(0)} \in \mathcal{C}$, to a fixed point of $f$ \cite[Theorem 5.1]{berinde}.

We notice that an SPC mapping is PC as well, therefore the Ishikawa iteration in \eqref{eq:Ishikawa} can be used in place of the Mann iteration in Corollary \ref{cor:iterations}. However, unlike the Mann iteration, in general there is no known convergence rate for the Ishikawa iteration, and in fact the convergence is usually much slower compared to the Mann iteration.
In this paper we have considered MF problems in which the aggregation mapping is at least SPC; an open question is whether there exist MF problems in which the aggregation mapping is PC, but not SPC, so that the Ishikawa iteration becomes necessary.

As exploited in the proofs of the main results, both SPC in Definition \ref{def:SPC} and PC in Definition \ref{def:PC} can be characterized in terms of monotone mappings, according to the following definitions and results \cite[Definition 1.14, p. 13]{berinde}, \cite[Definition 20.1]{bauschke:combettes}.

\vspace{0.2cm}
\begin{definition}[{Monotone mapping}] \label{def:MON}
A mapping $f: \R^n \rightarrow \R^n$ is strongly monotone (SMON) in $\mathcal{H}_P$ if there exists $\epsilon > 0$ such that
\begin{equation} \label{eq:SMON-inequality}
\left( f(x) - f(y) \right)^\top P \left( x-y\right) \geq \epsilon \left\| x-y\right\|_P^{2}
\end{equation}
for all $x, y \in \R^n$. It is monotone (MON) in $\mathcal{H}_P$ if \eqref{eq:SMON-inequality} holds with $\epsilon=0$.
{\hfill $\square$}
\end{definition}

\vspace{0.2cm}
\begin{lemma} \label{lem:MON-SAC}
If $f: \R^n \rightarrow \R^n$ is MON in $\mathcal{H}_P$ and $g: \R^n \rightarrow \R^n$ is SMON in $\mathcal{H}_P$, then $f+g$ is SMON in $\mathcal{H}_P$.
{\hfill $\square$}
\end{lemma}

\vspace{0.2cm}
\begin{proof}
It follows from Definition \ref{def:MON} that there exists $\epsilon>0$ such that, for all $x,y \in \R^n$:
\begin{equation*}
\begin{array}{l}
\left( f(x)+g(x) - \left( f(y) + g(y) \right) \right)^\top P \left( x-y\right) \\
\ = \left( f(x) - f(y) \right)^\top P \left( x-y\right) + \left( g(x) - g(y)\right)^\top P \left( x-y\right) \\
\ \geq \epsilon \left\| x-y\right\|_P^2.
\end{array}
\end{equation*}
\end{proof}

\vspace{0.2cm}
\begin{remark}
$f$ FNE $\Longrightarrow$ $f$ MON \cite[Example 20.5]{bauschke:combettes}; $f$ PC $\Longleftrightarrow$ $\text{Id} - f$ MON \cite[Example 20.8]{bauschke:combettes}.
{\hfill $\square$}
\end{remark}

\vspace{0.2cm}
\begin{lemma} \label{lem:from-SAC-to-SPC}
For any $f: \R^n \rightarrow \R^n$, the mapping $\text{Id} - f$ is SPC in $\mathcal{H}_P$ if and only if there exists $\epsilon > 0$ such that $ \left( f(x) - f(y) \right)^\top P \left( x-y\right) \geq \epsilon \left\| f(x) - f(y) \right\|_P^2$ for all $x, y \in \R^n$.
If $f$ is Lipschitz continuous and SMON in $\mathcal{H}_P$, then $\text{Id} - f$ is SPC in $\mathcal{H}_P$.
{\hfill $\square$}
\end{lemma}

\begin{proof}
By Definition \ref{def:SPC}, $\text{Id} - f$ is SPC if there exists $\rho < 1$ such that
$\left\| f(x) - f(y) - (x-y)\right\|_P^2 \leq \left\| x-y\right\|_P^2 + \rho \left\| f(x) - f(y)\right\|_P^2$ for all $x, y \in \R^n$. Equivalently, since 
$\left\| f(x) - f(y) - (x-y)\right\|_P^2  = \left\| f(x) - f(y) \right\|_P^2 + \left\| x-y\right\|_P^2 - 2 \left( f(x) - f(y) \right)^\top P (x-y) $, we have 
\begin{multline*}
\left\| f(x) - f(y)\right\|_P^2 - 2 \left( f(x) - f(y) \right)^\top P ( x-y ) \leq \rho \left\| f(x) - f(y)\right\|_P^2 \\
\Longleftrightarrow \ \textstyle \frac{1-\rho}{2} \left\| f(x) - f(y)\right\|_P^2 \leq \left( f(x) - f(y) \right)^\top P ( x-y )
\end{multline*}
for all $x, y \in \R^n$, which proves the first statement with $\epsilon = \frac{1-\rho}{2}$.
If $f$ is Lipschitz and SMON then there exist $L, \epsilon > 0$ such that
$ \epsilon \left\| f(x) - f(y) \right\|_P^2 \leq \epsilon L \left\| x - y\right\|_P^2 \leq L \left( f(x) - f(y) \right)^\top P ( x- y )$ for all $x, y \in \R^n$. Therefore, we have $\left( f(x) - f(y) \right)^\top P ( x- y ) \geq \frac{\epsilon}{L}  \left\| f(x) - f(y) \right\|_P^2$, which implies that $\text{Id} - f$ is SPC from the previous part of the proof.
\end{proof}

\subsection*{Regularity of affine mappings}
We next present necessary and sufficient conditions to characterize the regularity of affine mappings. 
Some of these equivalences are exploited in Appendix \ref{app:proofs}. The statements could be further exploited to show which fixed point iteration can be used to solve the unconstrained LQ deterministic MF control problem from Section \ref{sec:Discrete-time constrained linear quadratic mean field control}.

\vspace{0.2cm}
\begin{lemma}[Regularity of affine mappings] \label{lem:nonexpansive-affine}
The following equivalencies hold true for any mapping $f: \R^n \rightarrow \R^n$ defined as $f(x) := A x + b$, for some $A \in \R^{n \times n}$ and $b \in \R^n$. \vspace{0.1cm} \\
\begin{tabular}{lllc}
$1$. CON & in $\mathcal{H}_P$ & $\Longleftrightarrow$ & $A^\top P A - P \prec 0$ \\
$2$. NE & in $\mathcal{H}_P$ & $\Longleftrightarrow$ &   $A^\top P A - P \preccurlyeq 0$ \\
$3$. FNE & in $\mathcal{H}_P$ & $\Longleftrightarrow$ &  $2 A^\top P A \preccurlyeq A^\top P + P A $ \\
$4$. SMON & in $\mathcal{H}_P$ & $\Longleftrightarrow$ & $A^\top P + P A \succ 0$ \\
$5$. MON & in $\mathcal{H}_P$ & $\Longleftrightarrow$ &   $A^\top P + P A \succcurlyeq 0$ \\
$6$. PC & in $\mathcal{H}_P$ & $\Longleftrightarrow$ &   $A^\top P + P A \preccurlyeq 2 P$ \\
\end{tabular}

{\hfill $\square$}
\end{lemma}

\vspace{0.2cm}
\begin{proof}
The mapping $f$ is CON in $\mathcal{H}_P$ if and only if there exists $\epsilon \in (0,1]$ such that
$ \left\| f(x) - f(y)\right\|_P^2 = \left\| A(x-y)\right\|_P^2 \leq (1-\epsilon)^2 \left\| x-y\right\|_P^2$ for all $x,y \in \R^n$; \\ 
equivalently,
$\left( x-y\right)^\top A^\top P A \left( x-y\right) \leq (1- \epsilon)^2 \left( x-y\right)^\top P \left( x-y\right)$ for all $x,y \in \R^n$, that is $A^\top P A \preccurlyeq (1-\epsilon)^2 \, P \Leftrightarrow A^\top P A - P \preccurlyeq  -( 2 \epsilon - \epsilon^2 ) P$. Since $P \succ 0$, the existence of $\epsilon >0$ such that the latter matrix inequality holds is equivalent to the existence of $\varepsilon > 0$ such that $A^\top P A - P \preccurlyeq - \varepsilon I$.
An analogous proof with $\epsilon = \varepsilon = 0$ shows that the mapping $f$ is NE in $\mathcal{H}_P$ if and only if $A^\top P A - P\preccurlyeq 0$.

The mapping $f$ is FNE in $\mathcal{H}_P$ if and only if $ \left\| f(x) - f(y)\right\|_P^2 = \left\| A(x-y)\right\|_P^2 \leq \left\| x-y\right\|_P^2 -  \left\| f(x) - f(y) - (x-y)\right\|_P^2$ for all $x,y \in \R^n$. Equivalently, we get 
$ \left( x-y\right)^\top A^\top P A \left( x-y\right) \leq \left( x-y\right)^\top P \left( x-y\right) - \left( x-y\right)^\top \left( A - I\right)^\top P \left( A - I\right) \left( x-y\right) $ for all $x,y \in \R^n$, that is $A^\top P A \preccurlyeq P - (A-I)^\top P (A-I) = P - A^\top P A + A^\top P + P A - P \Leftrightarrow 2 A^\top P A \preccurlyeq A^\top P + P A$.

The mapping $f$ is SMON in $\mathcal{H}_P$ if and only if there exists $\epsilon > 0$ such that
$\left( f(x) - f(y) \right)^\top P (x-y) = (x-y)^\top A^\top P (x-y) \geq \epsilon \left\| x-y\right\|_P^2 = \epsilon (x-y)^\top P (x-y)$ for all $x,y \in \R^n$, that is equivalent to $\frac{1}{2} \left( A^\top P + P A \right) \succcurlyeq \epsilon P$. Since $P \succ 0$, the existence of $\epsilon >0$ such that the latter matrix inequality holds is equivalent to the existence of $\varepsilon > 0$ such that $A^\top P + P A \succcurlyeq \varepsilon I$.
An analogous proof with $\epsilon = \varepsilon = 0$ shows that the mapping $f$ is MON in $\mathcal{H}_P$ if and only if $A^\top P + P A \succcurlyeq 0$.

The mapping $f$ is PC in $\mathcal{H}_P$ if and only if $ \left\| f(x) - f(y)\right\|_P^2 = \left\| A(x-y)\right\|_P^2 \leq \left\| x-y\right\|_P^2 +  \left\| f(x) - f(y) - (x-y)\right\|_P^2 = \left\| x-y\right\|_P^2 +  \left\| (A-I)(x-y) \right\|_P^2$ for all $x,y \in \R^n$. Equivalently, we get
$\left( x-y\right)^\top A^\top P A \left( x-y\right) \leq \left( x-y\right)^\top P \left( x-y\right) + \left( x-y\right)^\top (I-A)^\top P (I-A) \left( x-y\right)$ for all $x,y \in \R^n$, that is $A^\top P A \preccurlyeq P + (I-A)^\top P (I-A) = 2P -(A^\top P + P A) + A^\top P A $ and hence $A^\top P + P A \preccurlyeq 2 P $.
\end{proof}

\subsection{Finite-horizon approximation of infinite-horizon discounted-cost optimal control problems} \label{app:finite-horizon}

Let us consider continuous and convex stage-cost functions $\{ \ell_t: \left(\times_{\tau=1}^{t} \mathcal{X}_{\tau}  \right) \rightarrow \R_{\geq 0} \}_{t=1}^{\infty}$, where for all $t \in \N$, $t\geq 1$, 
$\mathcal{X}_t \subseteq \mc{X} \subset \R^n$ is compact and convex, for some compact set $\mc{X}$.
Consider the infinite-dimensional set $\mathcal{S} := \times_{t=1}^{\infty} \mathcal{X}_t = \mathcal{X}_1 \times \mathcal{X}_2 \times \ldots$, $\beta \in (0,1)$, and the function $J_{\infty}: \mathcal{S} \rightarrow \R_{\geq 0}$
defined as
\begin{equation}\label{eq:J-inf-def}
\textstyle J_{\infty}\left( \{ x_t \}_{t=1}^{\infty}  \right) := \sum_{t=1}^{\infty} \beta^t \ell_t\left( \{ x_{\tau} \}_{\tau=1}^{t} \right).
\end{equation}
Let us also define
\begin{equation} \label{eq:x-inf-def}
J_{\infty}^{\star} := \inf_{ y \in \mathcal{S} } J_{\infty}(y), \quad x_{\infty}^{\star} := \arg\min_{ y \in \mathcal{S} } J_{\infty}(y),
\end{equation}
where we assume that the infimum $J_{\infty}^{\star}$ is attained in a unique point $x_{\infty}^{\star} \in \mathcal{S}$.

Analogously, let us define the finite-dimensional counterparts of the above quantities. We consider $\mathcal{S}_T := \times_{t=1}^{T} \mathcal{X}_t \subseteq \mathcal{X}^{T} \subset \left( \R^n \right)^T$, $J_T: \mathcal{S}_T \rightarrow \R_{\geq 0}$ defined as
\begin{equation}\label{eq:J-K-def}
\textstyle  J_{T}\left( \{ x_t \}_{t=1}^{T}  \right) := \sum_{t=1}^{T} \beta^t \ell_t\left( \{ x_{\tau} \}_{\tau=1}^{t} \right),
\end{equation}
besides the optimal value $J_T^{\star}$ and optimizer $x_T^{\star}$, assumed to be unique:
\begin{equation}\label{eq:x-K-def}
J_{T}^{\star} := \min_{ x \in \mathcal{S}_T } J_{T}(x), \quad x_{T}^{\star} := \arg\min_{ x \in \mathcal{S}_T } J_{T}(x).
\end{equation}

We next show that if $T$ is chosen large enough, then $J_T^{\star}$ gets arbitrarily close to $J_{\infty}^{\star}$.

\vspace{0.2cm}
\begin{proposition}[Finite-horizon approximation] \label{prop:finite-horizon-approximation}
Let $J_{\infty}^{\star}$, $J_{K}^{\star}$ be as in \eqref{eq:x-inf-def}, \eqref{eq:J-K-def}, respectively. Then
$\displaystyle \lim_{T \rightarrow \infty} | J_{T}^{\star} -  J_{\infty}^{\star} | = 0$.
{\hfill $\square$}
\end{proposition}

\begin{proof}
Let $\left( x_{\infty}^{\star} \right)_{t}$ denote the component $t$ of $x_{\infty}^{\star}$, which we rewrite as $x_{\infty}^{\star} = \left\{ \left( x_{\infty}^{\star} \right)_{t} \right\}_{t=1}^{\infty}$. We start from the following inequalities:
\begin{equation*}
\begin{array}{l}
J_{T}^{\star} \leq J_{T}\left( \left\{ \left( x_{\infty}^{\star} \right)_{t} \right\}_{t=1}^{T} \right) \leq J_{\infty}^{\star} \\
 \ = J_{T}\left( \left\{ \left( x_{\infty}^{\star} \right)_{t} \right\}_{t=1}^{T} \right) + \sum_{t = T+1}^{\infty} \beta^t \ell_t \left(  \{  \left( x_{\infty}^{\star} \right)_{\tau}  \}_{\tau=1}^{t} \right) \\
\ \leq  J_T^{\star} + \sum_{t = T+1}^{\infty} \beta^t \ell_t \left(  \{  y_\tau  \}_{\tau=1}^{t} \right)
\end{array}
\end{equation*}
where, for all $\tau \geq 1$, $y_\tau := \left( x_{T}^\star \right)_\tau \in \mathcal{X}_\tau$ if $\tau \in \Z[1,T]$, $y_t := \left( x_{\infty}^{\star} \right)_\tau \in \mathcal{X}_\tau$ if $\tau \geq T+1$. 
Now define $L := \sup_{t \geq 1 } \sup_{ \xi \in \mathcal{X}^t } \ell_t( \xi )$, and notice that $L < \infty$ as the functions $\{ \ell_t \}_{t \geq 1}$ are continuous and $\mathcal{X}$ is compact. We then have
\begin{equation*}
\textstyle 0 \leq J_{\infty}^{\star} - J_{T}^{\star} \leq L \sum_{t = T+1}^{\infty} \beta^t \leq \frac{L}{1-\beta} \beta^{T+1} \stackrel{T \rightarrow \infty}{\longrightarrow} 0,
\end{equation*}
from which we conclude that $ \lim_{T \rightarrow \infty} | J_{T}^{\star} -  \textstyle J_{\infty}^{\star} | = 0$.
\end{proof}

In presence of an exponential cost-discount factor in the cost as in \cite[Equation 2.2]{huang:caines:malhame:07}, \cite[Equation 2]{huang:caines:malhame:12}, Proposition \ref{prop:finite-horizon-approximation} suggests that a finite-horizon formulation can approximate a MF $\varepsilon$-Nash equilibrium relative to an infinite-horizon one. The formalization of such claim, under appropriate regularity conditions, is left as future work.

\subsection{Main proofs} \label{app:proofs}

We start from the characterization of the optimal solution of \eqref{eq:constrained-optimizer}.
\vspace{0.2cm}
\begin{lemma}[Parametric optimizer] \label{lem:parametric-optimizer}
The unconstrained optimizer in \eqref{eq:constrained-optimizer} is
\begin{equation} \label{eq:unconstrained-optimizer}
\hat{x}^{\star}(z) := \arg \min_{ x \in \mathbb{R}^n } J(x,z) = \left( Q + \Delta \right)^{-1} \left( (\Delta - C) z - c \right);
\end{equation}
the (constrained) optimizer in \eqref{eq:constrained-optimizer} reads as
\begin{equation} 
x^{i \, \star}(z) = \arg \min_{ x \in \mathcal{X}^i } J(x,z) = \text{Proj}_{\mathcal{X}^i }^{Q+\Delta}(\hat{x}^\star(z)).
\end{equation}
{\hfill $\square$}
\end{lemma}

\begin{proof}
The closed-form expression of the (unique) unconstrained optimizer $\hat x^\star(z)$ directly follows from the equation $0 = \frac{\partial}{ \partial x }J(x,z) = \frac{\partial}{ \partial x }\left( x^\top Q x +  (x-z)^\top \Delta (x-z) + 2 \left( C z + c \right)^\top x \right) =  2x^\top Q + 2(x-z)^\top \Delta + 2 \left( C z + c \right)^\top$.
Then the following equalities hold:
\begin{equation*}
\begin{array}{l}
\displaystyle \text{Proj}_{\mathcal{X}^i}^{Q + \Delta}(\hat{x}^\star(z)) = \arg \min_{ y \in \mathcal{X}^i } \left\| y -  \hat{x}^\star(z)\right\|_{Q + \Delta}^2  \\
\displaystyle = \arg \min_{ y \in \mathcal{X}^i } \left( y - ( Q + \Delta )^{-1}\left( (\Delta - C) z - c\right) \right)^\top (Q + \Delta) \cdot \\
\displaystyle \hfill \qquad \qquad \qquad \qquad \quad  \cdot \left( y - ( Q + \Delta )^{-1}\left( (\Delta - C) z - c\right) \right) \\
\displaystyle = \arg \min_{ y \in \mathcal{X}^i } y^\top (Q + \Delta) y - 2 y^\top \left( (\Delta - C) z - c\right) \\
\displaystyle = \arg \min_{ y \in \mathcal{X}^i } y^\top Q y + y^\top \Delta y - 2 y^\top \Delta z + 2 \left( C z + c \right)^\top y \\
\displaystyle = \arg \min_{ y \in \mathcal{X}^i } y^\top Q y + (y - z)^\top \Delta (y - z) + 2 \left( C z + c \right)^\top y \\
 = x^{i \, \star}(z).
\end{array}
\end{equation*}
\end{proof}

\vspace{0.2cm}
\begin{remark} \label{rem:optimizer-Lipschitz}
Since the mapping $\hat{x}^{\star}$ in \eqref{eq:unconstrained-optimizer} is affine and hence Lipschitz, and the projection operator $\text{Proj}_{\mathcal{X}^i}^{Q + \Delta}$ has Lipschitz constant $1$ in $\mc{H}_{Q+\Delta}$ \cite[Proposition 4.8]{bauschke:combettes}, both mappings $\hat{x}^{\star}(\cdot)$ and ${x}^{i \, \star}(\cdot) = \text{Proj}_{\mathcal{X}^i}^{Q + \Delta}( \hat{x}^{\star}(\cdot) )$ in \eqref{eq:constrained-optimizer} are Lipschitz with the same constant.
{\hfill $\square$}
\end{remark}

\subsection*{Proof of Proposition \ref{prop:exists-Nash}}
It follows from Definition \ref{def:Nash-equ} that a set of strategies $\{ \bar{x}^{i} \in \mc{X}^i \}_{i=1}^{N}$ is a MF Nash equilibrium if, for all $i \in \Z[1,N]$,
\begin{align*}
 \bar{x}^i &=\textstyle \argmin_{y \in \mc{X}^i} J\left( y, \frac{1}{N} a_i y + \frac{1}{N} \sum_{j \neq i}^{N} a_j \bar{x}^j \right) \\
  &= \textstyle { x^{i}_{ \textup{br} }({\boldsymbol{\bar{x}}^{-i}}) }=: \argmin_{y \in \mc{X}^i} \bar{J}\left( y,  \frac{1}{N} \sum_{j \neq i}^{N} a_j \bar{x}^j\right),
\end{align*}
where the cost function $\bar J$ is quadratic as well.
{Note that, for each $i \in \Z[1,N]$,} the quantity $\frac{1}{N} \sum_{j \neq i}^{N} a_j \bar{x}^j$ can be written as
\begin{equation*}
\begin{array}{l}
\textstyle \frac{1}{N} \left[ a_1 I_n, \ldots, a_{i-1} I_n, \, \bs{0}, \, a_{i+1} I_n, \ldots, a_N I_n \right] \left[ \bar{x}^1; \ldots; \bar{x}^N \right]  \vspace{0.05cm}\\
\textstyle = \left( \left[ \frac{a_1}{N}, \ldots, \frac{a_{i-1}}{N}, \, 0, \, \frac{a_{i+1}}{N}, \ldots, \frac{a_N}{N} \right] \otimes I_n \right) \left[ \bar{x}^1; \ldots; \bar{x}^N \right] \vspace{0.05cm} \\
\textstyle = \left( \boldsymbol{a}_{-i}^{\, \top} \otimes I_n \right) \left[ \bar{x}^1; \ldots; \bar{x}^N \right],
\end{array}
\end{equation*}
where we define $\boldsymbol{a}_{-i}^{\, \top} := \left[ \frac{a_1}{N}, \ldots, \frac{a_{i-1}}{N}, \, 0, \, \frac{a_{i+1}}{N}, \ldots, \frac{a_N}{N} \right] \in \R^{1 \times N}$ for all $i \in \Z[1,N]$. {Consequently,} 
$\{ \bar{x}^{i} \in \mc{X}^i \}_{i=1}^{N}$ is a MF Nash equilibrium if and only if
\begin{equation*}
\textstyle \left[ \begin{matrix} \bar{x}^1 \\ \vdots \\ \bar{x}^N \end{matrix} \right] = 
\textstyle \left[ \begin{matrix} \argmin_{y \in \mc{X}^1} \bar{J}\left( y,  \left( \boldsymbol{a}_{-1}^{\, \top} \otimes I_n \right) \left[ \begin{smallmatrix} \bar{x}^1 \\ \vdots \\ \bar{x}^N \end{smallmatrix} \right] \right)   \\ \vdots \\ 
\argmin_{y \in \mc{X}^N} \bar{J}\left( y,  \left( \boldsymbol{a}_{-N}^{\, \top} \otimes I_n \right) 
\left[ \begin{smallmatrix} \bar{x}^1 \\ \vdots \\ \bar{x}^N \end{smallmatrix} \right] \right) 
\end{matrix} \right].
\end{equation*}
Equivalently, $\{ \bar{x}^{i} \in \mc{X}^i \}_{i=1}^{N}$ is a MF Nash equilibrium if and only if $\left[ \bar{x}^1; \ldots; \bar{x}^N \right]$ is a fixed point of the continuous mapping
\begin{equation*}
 \left[ \begin{matrix} \argmin_{y \in \mc{X}^1} \bar{J}\left( y,  \left( \boldsymbol{a}_{-1}^{\, \top} \otimes I_n \right) \left( \bs{\cdot} \right) \right)   \\ \vdots \\ 
\argmin_{y \in \mc{X}^N} \bar{J}\left( y,  \left( \boldsymbol{a}_{-N}^{\, \top} \otimes I_n \right) \left( \bs{\cdot} \right) \right) 
\end{matrix} \right] 
\end{equation*}
from $\R^{nN}$ to the compact set $ \times_{i=1}^N \mc{X}^i \subset \R^{nN}$.
The existence of a fixed point of the latter mapping, and equivalently the existence of a MF Nash equilibrium, then follows from \cite[Theorem 4.1.5(b)]{smart1980fixed}.
\hfill{$\blacksquare$} 

\subsection*{Proof of Theorem \ref{thm:epsilon:nash}}

From Lemma \ref{lem:parametric-optimizer} we have that $x^{i \, \star}(z) = \text{Proj}_{\mathcal{X}^i}^{Q + \Delta}\left(  \left( Q + \Delta \right)^{-1} \left( (\Delta - C) z - c \right) \right)$, that is the metric projection (in the Euclidean space $\mathcal{H}_{Q + \Delta}$) onto the compact and convex set $\mathcal{X}^i$  of the affine mapping $z \mapsto \left( Q + \Delta \right)^{-1} \left( (\Delta - C) z - c \right)$. 
Therefore the mappings $\{x^{i \, \star} \}_{i=1}^{N}$  are Lipschitz with the same constant, that is, there exists $L>0$ such that $\left\| x^{i \, \star}(v) - x^{i \, \star}(w) \right\|_{\infty} \leq L \left\| v-w\right\|_{\infty}$ for all $v,w \in \R^n$ and for all $i \in \Z[1,N]$.

Now, $J$ in \eqref{eq:constrained-optimizer} is a quadratic function and takes values on a compact subset of $\R^n \times \R^n$, therefore it is Lipschitz, and hence there exists $M>0$ such that \\
$| J(v, z_1) - J(w,z_2) | \leq M \left( \left\| v-w \right\|_{\infty} + \left\| z_1-z_2 \right\|_{\infty} \right) $ for all $v,w \in \R^n$, $z_1, z_2 \in \R^n$. 
Let us also define $D := \max_{ v,w \in \mathcal{X} } \left\| v-w\right\|_{\infty}$, where $\mathcal{X} \supseteq \cup_{ N \geq 0 } \cup_{i=1}^{N} \mathcal{X}^i $ is compact from Assumption \ref{ass:uniform-compactness}.

We now consider an arbitrary fixed point $\bar{z} = \frac{1}{N}\sum_{i=1}^{N} a_i x^{i \, \star}(\bar{z}) =  \frac{1}{N}\sum_{i=1}^{N} a_i \bar{x}^i$ of the aggregation mapping $\mathcal{A}$ in \eqref{eq:mean}. 
We show that an arbitrary agent $i$ can improve its cost only by an amount $\varepsilon = \varepsilon_N = \mathcal{O}\left( 1/N \right)$ if we fix the strategies $\{ \bar{x}^j := x^{j \, \star}( \bar{z} ) \}_{j\neq i}^{N}$ of all other agents. 
Let $\tilde{x}^{i \, \star}$ denote the optimal strategy for agent $i$ given the strategies of the others $\{ \bar{x}^j \}_{j\neq i}^{N}$, that is, 
$\tilde{x}^{i \, \star} := \arg\min_{y \in \mathcal{X}^i } J\left( y, \frac{1}{N}\left(  a_i y + \sum_{j \neq i}^{N} a_j \bar{x}^j \right) \right) $, and let 
$\displaystyle \tilde{\tilde{x}}^{i \, \star} := \arg\min_{y \in \mathcal{X}^i } \textstyle J\left( y, \frac{1}{N}\left(  a_i \tilde{x}^{i \, \star} + \sum_{j \neq i}^{N} a_j \bar{x}^j \right) \right) $.
Let us also define the associated costs:
\begin{equation*}
\begin{array}{lll}
\textstyle \bar{J}^{i \, \star} & = &
\textstyle J\left( \bar x^{i }, \frac{1}{N}\left(  a_i \bar x^{i } + \sum_{j \neq i}^{N} a_j \bar{x}^j \right) \right) \\
& = & \textstyle\min_{y \in \mathcal{X}^i } J\left( y, \frac{1}{N}\left(  a_i \bar x^i + \sum_{j \neq i}^{N} a_j \bar{x}^j \right) \right), \\
\textstyle \tilde{J}^{ i \, \star} & = &  
\textstyle J\left( \tilde x^{i \, \star }, \frac{1}{N}\left(  a_i \tilde x^{i \, \star } + \sum_{j \neq i}^{N} a_j \bar{x}^j \right) \right) \\ 
& = & \textstyle
\min_{y \in \mathcal{X}^i } J\left( y, \frac{1}{N}\left(  a_i y + \sum_{j \neq i}^{N} a_j \bar{x}^j \right) \right), \\
\textstyle \tilde{\tilde{J}}^{ i \, \star} & = &
\textstyle J\left( \tilde{\tilde{x}}^{ i \, \star }, \frac{1}{N}\left(  a_i \tilde x^{ i \, \star } + \sum_{j \neq i}^{N} a_j \bar{x}^j \right) \right) \\
& = & \textstyle \min_{y \in \mathcal{X}^i } J\left( y, \frac{1}{N}\left(  a_i \tilde x^{ i \, \star} + \sum_{j \neq i}^{N} a_j \bar{x}^j \right) \right).
\end{array}
\end{equation*}
Note that $\tilde{\tilde{J}}^{i \, \star} \leq \tilde{J}^{ i \, \star} \leq\bar{J}^{ i \, \star}$. 
Then we define $\tilde{z} := \frac{1}{N}\left(  a_i \tilde{x}^{i \, \star} +  \sum_{j \neq i}^{N} a_j\bar{x}^j \right)$ and notice that 
$\bar{x}^i = x^{ i \, \star}\left( \bar{z}\right)$ and 
$\tilde{\tilde{x}}^{i \, \star} = x^{i \, \star}\left( \tilde{z}\right)$.
Therefore, the following inequalities hold true:
\begin{align}
\begin{split}
\textstyle 0 & \leq \bar{J}^{ i \, \star} - \tilde{J}^{i \, \star}  \leq \bar{{J}}^{ i \, \star} - \tilde{\tilde{J}}^{ i \, \star}  = 
\left| J\left( \bar{x}^i, \bar{z} \right) - J\left( \tilde{\tilde{{x}}}^{i \, \star}, \tilde{z} \right) \right|  \\
& \leq \textstyle M \left\| \bar{x}^i - \tilde{\tilde{x}}^{i \, \star} \right\|_\infty  + \textstyle M \left\| \bar{z} - \tilde{z} \right\|_\infty  \\
& = \textstyle M \left\| \bar{x}^i - \tilde{\tilde{x}}^{i \, \star} \right\|_\infty + 
\textstyle \frac{M}{N} a_i \left\| \bar{x}^i - \tilde{x}^{i \, \star} \right\|_\infty  \\
& = M \left\| x^{i \, \star}\left( \bar{z}\right) - x^{i \, \star}\left( \tilde{z}\right) \right\|_\infty + 
\textstyle \frac{M}{N} a_i \left\| \bar{x}^i - \tilde{x}^{i \, \star} \right\|_\infty  \vspace{0.1cm} \\
& \leq  M \, L \left\| \bar{z} - \tilde{z} \right\|_\infty +  \frac{M}{N} \bar{a} \left\| \bar{x}^i - \tilde{x}^{i \, \star} \right\|_\infty \\
& = \textstyle  \frac{\bar{a} \, M \, (L+1)}{N} \left\| \bar{x}^i - \tilde{x}^{i \, \star} \right\|_\infty \\
& \leq \textstyle  \frac{\bar{a} \, M \, D \, (L+1)}{N} =: \varepsilon_N.
\end{split}
\end{align}
This proves that for all $\varepsilon > 0$ there exists $\bar{N}_{\varepsilon} := \frac{\bar{a} \, M \, D\, (L+1)}{\varepsilon}$ such that the cost $\bar{J}^{ i \, \star}$ of any agent $i$ at a fixed point $\bar{z}$ is $\varepsilon$-close to its true optimal cost $\tilde{J}^{i \, \star}$, for all population sizes $N \geq \bar{N}_{\varepsilon}$.
{\hfill $\blacksquare$}

\subsection*{Proof of Theorem \ref{th:mean-pseudocontractive}}

It follows from the proof of Lemma \ref{lem:nonexpansive-affine} in Appendix \ref{app:operator-theory} that the unconstrained optimizer $\hat{x}^\star$ in \eqref{eq:unconstrained-optimizer} is CON in $\mathcal{H}_{Q+\Delta}$ if and only if there exist $\epsilon > 0$ such that
$ \left( (Q+\Delta)^{-1}(\Delta - C) \right)^\top (Q + \Delta) \left( (Q+\Delta)^{-1}(\Delta - C) \right) = (\Delta - C)^\top  (Q+\Delta)^{-1}(\Delta - C) \preccurlyeq (1-\epsilon)^2 (Q+\Delta)  \Leftrightarrow 
(\Delta - C)^\top  \left( (1-\epsilon) \, (Q+\Delta) \right)^{-1}(\Delta - C) \preccurlyeq (1-\epsilon) \, (Q+\Delta)$. 

As $Q + \Delta \succ 0$, by the Schur complement \cite[Section A.5.5]{boyd:vandenberghe} the last inequality is equivalent to 
\begin{equation*}
\begin{array}{l}
\left[ \begin{matrix}  (1-\epsilon) (Q + \Delta) & \Delta - C \\ (\Delta - C)^\top &  (1-\epsilon)(Q + \Delta) \end{matrix}\right] \succcurlyeq 0
\  \vspace{0.2cm} \\
\Leftrightarrow \left[ \begin{matrix}  Q + \Delta & \Delta - C \\ (\Delta - C)^\top & Q + \Delta \end{matrix}\right] \succcurlyeq \epsilon  \left[ \begin{matrix}  Q + \Delta & 0 \\ 0 & Q + \Delta \end{matrix}\right] \  \vspace{0.25cm} \\
\Leftrightarrow \left[ \begin{matrix}  Q + \Delta & \Delta - C \\ (\Delta - C)^\top & Q + \Delta \end{matrix}\right] \succcurlyeq \varepsilon I_{2n}
\end{array}
\end{equation*}
for some $\varepsilon > 0$.
The proof that $\hat{x}^\star$ in \eqref{eq:unconstrained-optimizer} is NE in $\mathcal{H}_{Q+\Delta}$ if and only if  \eqref{eq:condition-contractive-constrained-optimizer} holds with $\epsilon \geq 0$ is analogous (with $\epsilon = \varepsilon = 0$).

Since $\textup{Proj}_{\mathcal{X}^i}^{Q + \Delta}$ is FNE \cite[Proposition 4.8]{bauschke:combettes} and hence NE in $\mathcal{H}_{Q + \Delta}$, that is \\
$\textstyle \left\| \textup{Proj}_{\mathcal{X}^i}^{Q + \Delta}(x) - \textup{Proj}_{\mathcal{X}^i}^{Q + \Delta}(y) \right\|_{Q + \Delta} \leq \left\| x - y \right\|_{ Q + \Delta }$ for all $x, y \in \R^n$, it follows that the composition  $x^{i \, \star}(\cdot) = \textup{Proj}_{\mathcal{X}^i}^{ Q + \Delta }( \hat{x}^{\star}(\cdot) )$ is CON in $\mathcal{H}_{Q + \Delta}$ if $\hat x^\star$ is CON in $\mathcal{H}_{Q + \Delta}$,  NE in $\mathcal{H}_{Q + \Delta}$ if $\hat x^\star$ is NE in $\mathcal{H}_{Q + \Delta}$.

For the rest of the proof, we need the following fact, adapted from \cite[Proposition 4.2 (iv)]{bauschke:combettes}.

\vspace{0.2cm}
\begin{lemma} \label{lem:equivalent-FNE}
A mapping $f: \R^n \rightarrow \R^n$ is FNE in $\mathcal{H}_{P}$, $P \in \bbS_{\succ 0}^{n}$, if and only if
\begin{equation} \label{eq:FNE-equivalent-condition}
\left\| f(x) - f(y) \right\|_{P}^{2} \leq \left( x-y\right)^\top P \left( f(x) - f(y)\right)
\end{equation}
for all $x,y \in \R^n$.
{\hfill $\square$}
\end{lemma}

\vspace{0.2cm}
\begin{proof}
From Definition \ref{def:FNE-mapping}, we have $f$ FNE if and only if, for all $x, y \in \R^n$,
\begin{equation*}
\begin{array}{l}
\left\| f(x) - f(y) \right\|_{P}^{2} \leq \left\| x-y\right\|_P^{2} - \left\| (x-y) - \left( f(x) - f(y)\right) \right\|_P^{2} \\
= \left\| x-y\right\|_P^{2} - \\
\left( \left\| x-y\right\|_P^{2} + \left\| f(x) - f(y) \right\|_{P}^{2} - 2 \left( x-y\right)^\top P \left( f(x) - f(y)\right) \right) \\
= - \left\| f(x) - f(y) \right\|_{P}^{2}  + 2 \left( x-y\right)^\top P \left( f(x) - f(y)\right),
\end{array}
\end{equation*}
and equivalently \eqref{eq:FNE-equivalent-condition}.
\end{proof}

From \cite[Proposition 4.8]{bauschke:combettes} we have that $\text{Proj}_{ \mathcal{C}}^{ P }$ is FNE  in $\mathcal{H}_{P}$,  hence by Lemma \ref{lem:equivalent-FNE}:
\begin{equation}\label{eq:FNE-distorted-projection}
\left\| \text{Proj}_{ \mathcal{C}}^{ P }( \tilde{x} ) - \text{Proj}_{ \mathcal{C}}^{ P }( \tilde{y} )\right\|_{P}^{2} \leq \\
\left( \tilde{x} - \tilde{y} \right)^\top P \left( \text{Proj}_{ \mathcal{C}}^{ P }( \tilde{x} ) - \text{Proj}_{ \mathcal{C}}^{ P }( \tilde{y} )\right)
\end{equation}
for all $\tilde{x}, \tilde{y} \in \R^n$. Therefore, with $\tilde{x} := A x + b$ and $ \tilde{y} := A y + b$, the FNE condition in \eqref{eq:FNE-distorted-projection} implies that
\begin{equation} \label{eq:FNE-nondistorted-projection}
\left\| \text{Proj}_{ \mathcal{C}}^{ P }(Ax+b) - \text{Proj}_{ \mathcal{C}}^{ P }(Ay+b)\right\|_{P}^{2} \leq \\
\left( x-y\right)^\top A^\top P \left( \text{Proj}_{ \mathcal{C}}^{ P }(Ax+b) - \text{Proj}_{ \mathcal{C}}^{ P }(Ay+b)\right)
\end{equation}
for all $x,y \in \R^n$.


Now, since $x^{i \, \star}(z) = \text{Proj}_{ \mathcal{X}^i }^{Q + \Delta}\left( \hat{x}(z)\right) = \text{Proj}_{ \mathcal{X}^i }^{Q + \Delta}\left( (Q + \Delta)^{-1}\left(  (\Delta - C) z - c \right)  \right)$ from \eqref{eq:constrained-optimizer}, let us consider \eqref{eq:FNE-nondistorted-projection} with $Q+\Delta$ in place of $P$, $(Q + \Delta)^{-1} (\Delta - C)$ in place of $A$, $-(Q + \Delta)^{-1}c$ in place of $b$,
$\mathcal{X}^i$ in place of $\mathcal{C}$, and $v,w$ in place of $x,y$. We hence obtain
\begin{equation} \label{eq:FNE-nondistorted-projection-xn}
0 \leq \left\| x^{i \, \star}(v) - x^{i \, \star}(w) \right\|_{Q+\Delta}^{2} \leq \\
\left( v-w\right)^\top (\Delta - C)^\top \left( x^{ i \, \star}(v) - x^{ i \, \star}(w)\right)
\end{equation}
for all $v,w \in \R^n$. 

If $Q+\Delta \succcurlyeq \Delta - C \succ 0$, i.e., $-Q \preccurlyeq C \prec \Delta$, then $\left\| x^\star(v) - x^\star(w) \right\|_{\Delta - C}^{2} \leq \left\| x^\star(v) - x^\star(w) \right\|_{Q+\Delta}^{2} $ for all $v,w \in \R^n$. Therefore, it follows from \eqref{eq:FNE-nondistorted-projection-xn} that
\begin{equation*} 
\left\| x^\star(v) - x^\star(w) \right\|_{\Delta - C}^{2} \leq \left( v-w\right)^\top (\Delta - C) \left( x^\star(v) - x^\star(w)\right)
\end{equation*}
for all $v,w \in \R^n$, which is equivalent to $x^\star$ being FNE in $\mathcal{H}_{\Delta - C}$ by Lemma \ref{lem:equivalent-FNE}. 

On the other hand, from \eqref{eq:FNE-nondistorted-projection-xn} we get
\begin{equation*}
0 \leq \left(  x^\star(w) - x^\star(v) \right)^\top (C - \Delta) \left( v-w\right) 
\end{equation*}
for all $v,w$, which for $C - \Delta \succ 0$ is equivalent to $-x^\star(\cdot)$ being MON in $\mathcal{H}_{C-\Delta}$ by Definition \ref{def:MON}. We now notice that $\text{Id}(\cdot)$ is a SMON mapping by Definition \ref{def:MON}; hence $\text{Id} - x^\star$, sum of SMON and MON mappings, is SMON in $\mathcal{H}_{C-\Delta}$ by Lemma \ref{lem:MON-SAC}. It then follows from Lemma \ref{lem:from-SAC-to-SPC} in Appendix \ref{app:operator-theory} that $\text{Id} - x^\star$ Lipschitz continuous and SMON in $\mathcal{H}_{C-\Delta}$ implies that $\text{Id} - \left( \text{Id} - x^\star \right) = x^\star$ is SPC in $\mathcal{H}_{C-\Delta}$.
{\hfill $\blacksquare$}

\subsection*{Proof of Theorem \ref{prop:mean-pseudocontractive}}
The mapping $\mathcal{A}$ in \eqref{eq:mean} is a convex hull among the mappings $\{ x^{i \, \star} \}_{i=1}^{N}$, that are uniformly Lipschitz in view of Remark \ref{rem:optimizer-Lipschitz}, therefore $\mathcal{A}$ is Lipschitz continuous as well. Moreover, $\mathcal{A}$ is compact valued on $\frac{1}{N} \sum_{i=1}^{N} a_i \mc{X}^i$, thus it has at least one fixed point \cite[Theorem 4.1.5(b)]{smart1980fixed}.

It follows from Theorem \ref{th:mean-pseudocontractive} that if $-Q \preccurlyeq C \prec \Delta$ then, for all $i \in \Z[1,N]$, the mapping $x^{i \, \star}(\cdot)$ is FNE in $\mathcal{H}_{\Delta - C}$. Therefore, $\mathcal{A}(\cdot) = \frac{1}{N}\sum_{i=1}^{N} a_i x^{i \, \star}(\cdot) $, convex combination of FNE mappings, is FNE as well \cite[Example 4.31]{bauschke:combettes}. Analogously, the convex combination of CON (NE) mappings is CON (NE) as well.

For the SPC case, if $\Delta \prec C$ then it follows from the proof of Theorem \ref{th:mean-pseudocontractive} that, for all $i \in \Z[1,N]$, $\text{Id} - x^{i \, \star}$ is SMON in $\mathcal{H}_{C-\Delta}$, see Definition \ref{def:MON}. Then it follows from Lemma \ref{lem:MON-SAC} that 
$\frac{1}{N} \sum_{i=1}^{N} a_i \left( \text{Id}(\cdot) - x^{i \, \star}(\cdot) \right)$ is SMON as well, which implies that 
$\text{Id} -  \frac{1}{N} \sum_{i=1}^{N} \left\{ a_i \text{Id} - a_i x^{i \, \star} \right\} = \frac{1}{N} \sum_{i=1}^{N} a_i x^{i \, \star} = \mathcal{A}$ is SPC in view of Lemma \ref{lem:from-SAC-to-SPC}.
{\hfill $\blacksquare$}

\subsection*{ Proof of Corollary \ref{cor:iterations}}
From Theorem \ref{th:mean-pseudocontractive}, if \eqref{eq:condition-contractive-constrained-optimizer} holds for some $\epsilon > 0$, then $\mathcal{A}$ in \eqref{eq:mean} is CON and if $-Q \preccurlyeq C \prec \Delta$, then $\mathcal{A}$ is FNE. In both cases, the Picard--Banach iteration converges a fixed point of $\mathcal{A}$ \cite[Theorem 2.1]{berinde}, \cite[Section 1, p. 522]{combettes:pennanen:02}, which is unique if $\mathcal{A}$ is CON.

For the other two fixed point iterations, we need to consider $\mathcal{A}$ in \eqref{eq:mean} as a mapping from a compact convex set to itself. 
This can be assumed without loss of generality (that is, up to discarding the initial condition $z_{(0)}$) since $\mathcal{A}$ takes values in $\frac{1}{N} \sum_{i=1}^{N} a_i \mathcal{X}^{i}$, which is a linear transformation of the compact convex sets $\{ \mathcal{X}^i \}_{i=1}^{N}$, as hence compact and convex as well \cite[Section 3, Theorem 3.1]{rockafellar}.
If \eqref{eq:condition-contractive-constrained-optimizer} holds for some $\epsilon \geq 0$ then $\mathcal{A}$ is NE from Theorem \ref{th:mean-pseudocontractive} and the Krasnoselskij iteration converges to a fixed point of $\mathcal{A}$ \cite[Theorem 3.2]{berinde}. 

Finally, if $\epsilon \geq 0$ in \eqref{eq:condition-contractive-constrained-optimizer} or $\Delta \prec C$ hold true, then $\mathcal{A}$ is SPC. Therefore the Mann iteration converges to a fixed point \cite[Fact 4.9, p. 112]{berinde}, \cite[Theorem R, Section I]{osilike:udomene:01}.
{\hfill $\blacksquare$}

\subsection*{Proof of Corollary \ref{cor:Huang-converges}}
It follows from Section \ref{sec:Discrete-time constrained linear quadratic mean field control} that the LQ optimal control problem in \eqref{eq:xi-ui-MFC} with cost function $J_{\gamma}$ in \eqref{eq:J-x-u-z-g}, can be rewritten in the same format of \eqref{eq:constrained-optimizer} with block structured matrices
\begin{equation*}
Q= \text{diag}(\bs{0}, \tilde R), \quad \Delta = \text{diag}( \tilde{Q}, \bs{0}), \quad C=(1-\gamma)\Delta,
\end{equation*}
where $\tilde R:=\textup{diag}(R_0,\hdots,R_{T-1})\succ 0$ and $\tilde Q:=\text{diag}(Q_1,\hdots,Q_{T})\succ 0$. 
To exploit the first point in Corollary~\ref{cor:iterations}, we need to consider the matrix
\begin{equation*}
\begin{array}{l}
\left[\begin{array}{cc}Q+\Delta & \Delta-C \\(\Delta-C)^\top &  Q+\Delta\end{array}\right]=
\left[\begin{array}{c|c} \left.\begin{smallmatrix}\tilde Q & 0 \\ 0 & \tilde R\end{smallmatrix}\right. & \left.\begin{smallmatrix}\gamma \tilde Q & 0 \\ 0 & 0\end{smallmatrix}\right.   \vspace{0.05cm} \\  \hline \noalign{\vskip 0.025cm} \left.\begin{smallmatrix} \gamma \tilde Q & 0 \\ 0 & 0\end{smallmatrix}\right.  &  \left.\begin{smallmatrix}\tilde Q & 0 \\ 0 & \tilde R\end{smallmatrix}\right.\end{array}\right] \\
\ = \Pi^\top \text{diag}\left( \left[\begin{smallmatrix}1 & \gamma  \\ \gamma   &1\end{smallmatrix}\right] \otimes \tilde Q \,, I_2 \otimes \tilde R \right) \Pi,
\end{array}
\end{equation*}
where $\Pi \in \R^{ 2(p+m)T \, \times \, 2(p+m)T }$ is the permutation matrix that swaps the second and third block columns. 
Since the eigenvalues of the Kronecker product of two matrices equal to the product of the eigenvalues of the two matrices, we have that $I_2 \otimes \tilde R \succ 0$ and that $\left[\begin{smallmatrix}1 & \gamma  \\ \gamma   &1\end{smallmatrix}\right] \otimes \tilde Q $ is positive definite if $-1 < \gamma<1$, positive semidefinite if $-1\le \gamma \leq 1$. Since $\Pi$ is invertible ($\Pi^\top \Pi = I$) and hence has no $0$ eigenvalues, we conclude that $\Pi^\top \text{diag}\left( \left[\begin{smallmatrix}1 & \gamma  \\ \gamma   &1\end{smallmatrix}\right] \otimes \tilde Q \,, I_2 \otimes \tilde R \right) \Pi \succ 0$ 
($\succcurlyeq 0$) if $-1 < \gamma<1$ ($-1 \leq \gamma \leq 1$). The proof then follows from Corollary~\ref{cor:iterations}.
{\hfill $\blacksquare$}

\subsection*{Proof of Corollary \ref{cor:Callaway-converges}}
We consider the matrix inequality \eqref{eq:condition-contractive-constrained-optimizer} in Theorem \ref{th:mean-pseudocontractive} with $Q = 0$, $\Delta = \delta I$, $\delta>0$, and $C = a I$, $a>0$.  The existence of $\epsilon > 0$ such that
$$\left[ \textstyle  \begin{matrix} \delta I & (\delta - a) I \\ (\delta-a)I & \delta I \end{matrix} \right] \succcurlyeq \epsilon I,$$
is equivalent, by Schur complement \cite[Section A.5.5]{boyd:vandenberghe}, to 
$\delta - (\delta - a) \delta^{-1} (\delta - a) > 0 \Leftrightarrow \delta^2 - (\delta - a)^2 > 0 \Leftrightarrow \delta > a/2$. This implies that if $\delta > a/2$ then $\mathcal{A}$ is CON in $\mathcal{H}_{\delta I }$ and, from Corollary \ref{cor:iterations}, the Picard--Banach iteration in \eqref{eq:Picard-Banach} converges to its unique fixed point.

We now consider the case of $\delta = a/2$. The condition of Theorem \ref{th:mean-pseudocontractive} for $\mathcal{A}$ being NE in $\mathcal{H}_{\delta I }$ is that $\frac{a}{2} \left[ \begin{smallmatrix} I & -I \\ -I & I \end{smallmatrix} \right] = \frac{a}{2} \left[ \begin{smallmatrix} 1 & -1 \\ -1 & 1\end{smallmatrix}\right] \otimes I \succcurlyeq 0$, which is satisfied because $a>0$ and $\left[ \begin{smallmatrix} 1 & -1 \\ -1 & 1\end{smallmatrix}\right] \otimes I$ has non-negative eigenvalues.
The convergence of the Krasnoselskij iteration in \eqref{eq:Krasnoselskij} follows from Corollary \ref{cor:iterations}.

We finally consider the case $\delta \in (0, \, a/2)$. From the sufficient condition in Theorem \ref{th:mean-pseudocontractive}, we get that $\mathcal{A}$ is SPC in $\mathcal{H}_{(a-\delta)I }$ if $\delta \in (0, \, a)$. The convergence of the Mann iteration in \eqref{eq:Mann} then follows from Corollary \ref{cor:iterations}.
{\hfill $\blacksquare$}

\bibliographystyle{IEEEtran}
\bibliography{library}

\end{document}